\UseRawInputEncoding
\documentclass[twoside,journal]{IEEEtran}
\usepackage{bbm}
\usepackage{tipa}
\usepackage{mathrsfs}
\usepackage{algorithmic}
\usepackage{algorithm}
\usepackage[dvips]{graphicx}
\usepackage{amsfonts}
\usepackage{amssymb}
\usepackage{amsmath}
\usepackage{mathrsfs}
\usepackage{verbatim}
\usepackage{subfigure}
\usepackage{balance}
\usepackage{booktabs}
\usepackage{fancybox}
\usepackage{bm}
\usepackage{pifont}
\usepackage{extarrows}
\usepackage{cite}
\usepackage{color}
\usepackage{cases}
\usepackage{rotating}
\newcounter{TempEqCnt}
\newtheorem{theorem}{Theorem}
\newtheorem{lemma}{Lemma}

\newtheorem{proposition}{Proposition}
\newtheorem{proof}{Proof}
\begin{document}
\bibliographystyle{IEEEtran}

\title{Performance Analysis of Joint Active User Detection and Channel Estimation for Massive Connectivity
%\thanks{\scriptsize
%This work was supported in part by the National Key R\&D Program of China under Grant 2019YFE0113200, in part by the National Natural Science Foundation of China under Grants 61941105 and 62171364. (\emph{Corresponding author: Hui-Ming Wang.})}
\author{Jia-Cheng Jiang, \hspace{0.02in}Hui-Ming Wang, \emph{Senior Member}, \emph{IEEE}\hspace{0.02in} and H. Vincent Poor, \emph{Life Fellow}, \emph{IEEE}\hspace{0.02in}}
\thanks{Jia-Cheng Jiang and Hui-Ming Wang are with the
School of Information and Communications Engineering, Xi'an Jiaotong
University, Xi'an, 710049, Shaanxi, P. R. China, and also with the Ministry
of Education Key Lab for Intelligent Networks and Network Security, Xi'an, 710049, Shaanxi, P. R. China. (e-mail: j1143484496b@stu.xjtu.edu.cn; xjbswhm@gmail.com.)}
\thanks{H. Vincent Poor is with the Department of Electrical Engineering, Princeton University, Princeton, NJ 08544 USA (e-mail: poor@princeton.edu)}
}
\maketitle
\begin{abstract}
This paper considers joint active user detection (AUD) and channel estimation (CE) for massive connectivity scenarios with sporadic traffic. The state-of-art method under a Bayesian framework to perform joint AUD and CE in such scenarios is approximate message passing (AMP). However, the existing theoretical analysis of AMP-based joint AUD and CE can only be performed with a given fixed point of the AMP state evolution function, lacking the analysis of AMP phase transition and Bayes-optimality. In this paper, we propose a novel theoretical framework to analyze the performance of the joint AUD and CE problem by adopting the replica method in the Bayes-optimal condition. Specifically, our analysis is based on a general channel model, which reduces to particular channel models in multiple typical MIMO communication scenarios. Our theoretical framework allows ones to measure the optimality and phase transition of AMP-based joint AUD and CE as well as to predict the corresponding performance metrics under our model. To reify our proposed theoretical framework, we analyze two typical scenarios from the massive random access literature, i.e., the isotropic channel scenario and the spatially correlated channel scenario. Accordingly, our performance analysis produces some novel results for both the isotropic Raleigh channel and spatially correlated channel case.
\end{abstract}
\section{Introduction}
%Massive access, also known as massive connectivity
%or massive machine-type communication (mMTC), is one of the main use cases of the fifth-generation (5G) and beyond 5G (B5G) wireless networks. Different from the traditional communication networks that aim to provide high data rate for relatively small numbers of user devices, the core mission of mMTC is to provide connectivity for millions of devices with low transmission rates \cite{IoTSur2015}. Another key characteristic of mMTC is that most devices are designed to sleep most of the time for energy efficiency and are only activated when triggered by external events. This feature causes the traffic patterns for user devices to be partly unpredictable and sporadic as only a subset of users are active concurrently \cite{BockelmannMassive2016}.

Recently, under the basic characteristics of mMTC, i.e, a large number of user devices and sporadic user traffic, the grant-free access strategy has been considered to allow the active devices to access the wireless network without a grant \cite{GraFrA2017, SenelGrant-Free2018}, which reduces both the access latency and signal processing overhead. Accordingly, the base station (BS) should simultaneously identify all the active user devices under grant-free random access. Furthermore, the BS is required to accurately acquire the channel state information for decoding uplink signals and executing downlink precoding after user activity detection. Hence, both active user detection (AUD) and channel estimation (CE) are required at the BS based on pilot sequences sent by the user devices.
Typically, the number of user devices in mMTC is very large, and thus assigning orthogonal pilot sequences to all user devices is prohibitive in practice, motivating the application of non-orthogonal pilot sequences in such situations. A central problem in the mMTC scenario is to jointly perform AUD and CE based on non-orthogonal pilot sequences.

Different from the sparse detection problem \cite{7563406, 7885116}, the joint AUD and CE requires both the support set and the corresponding amplitudes of the sparse vector, and can be formulated as a compressed sensing (CS) problem. Several CS-based solutions have been reported \cite{DiRe2020, LiuMassive12018, MRAcoCH2020, EPJAUDCE2019,jiang2021}. In \cite{DiRe2020}, a dimension reduction
method to reduce the pilot sequence length and computational complexity for joint AUD and CE has been proposed, which projects the original device state matrix onto a low-dimensional space by exploiting its sparse and low-rank structure. To fully utilize the statistical information of wireless channels, some Bayesian-based CS methods have been designed to achieve joint AUD and CE with higher accuracy. Specifically, the authors in \cite{LiuMassive12018} proposed an approximate message passing (AMP)-based joint AUD and CE strategy based on the randomly generated non-orthogonal pilot sequences for a massive MIMO system, and the statistical knowledge of both the channel and user sparsity was modeled using a prior distribution to facilitate the performance of AMP. Along this line, the authors in \cite{MRAcoCH2020} design an AMP-based algorithm to adaptively detect the active devices by exploiting the virtual
angular domain sparsity of the channels in an orthogonal frequency division multiplexing (OFDM) broadband system. Furthermore, an expectation propagation (EP)-based joint AUD and CE algorithm was proposed in \cite{EPJAUDCE2019} for massive access with a single-antenna BS, where the computationally intractable posterior probability of the involved sparse signals was approximated by a multivariate Gaussian distribution, enhancing the performance of joint AUD and CE. In addition, we have claimed in \cite{jiang2021} that the inherent temporal correlations of both the user channels and the active indicators between adjacent time slots can be used to enhance both the detection performance and channel estimation performance.
%way for joint aud and ce

One critical issue for joint AUD and CE for massive connectivity is to analyze the performance of user detection and channel estimation. The AMP algorithm provides a corresponding theoretical framework called state evolution to accurately track the performance of AMP in each iteration. By using state evolution analysis, the authors in \cite{LiuSparse2018, LiuMassive12018} provide analytical characterizations of the missed detection and false alarm probabilities for device detection and channel error covariance for channel estimation under a massive MIMO Raleigh channel assumption. However, although state evolution provides a theoretical framework for analyzing the performance metrics of joint AUD and CE, there are some critical issues that are left untouched. First, the performance analysis is based on the fixed point of the state evolution function, but the connection between the performance and the specific choices of system parameters such as pilot length, transmit power, fraction of number of active users and number of antennas was not established. Second, for specific choices of system parameters, the AMP iterations are blocked in a sub-optimal fixed point, so that the Bayes-optimal AUD and CE performance cannot be achieved via the AMP framework \cite{2012Probabilistic}. The region of the system parameters where the AMP framework is sub-optimal cannot be measured via the state evolution analysis in \cite{LiuSparse2018, LiuMassive12018}. Finally, some related works \cite{RepMMVzhu2018, RepMMVG2018} have shown that the performance of AMP exhibits a \emph{phase transition} phenomenon where the AMP algorithm will exhibit disconnect performance variations with variations of the system parameters, which has not been analyzed for massive connectivity.

%involving a set of typical replica assumptions that include a self-averaging property, the validity of ``replica trick'', the ability to exchange certain limits and the replica symmetry.

To address the above problems, in this paper, we propose a theoretical framework to analyze the performance of joint AUD and CE based on the replica method \cite{RepTan2002, 1985Entropy}, which is a classical tool for analyzing large systems that comes primarily from physics. The proposed framework established the connection between the system parameters and the performance of joint AUD and CE. Based on that, both the Bayes-optimal and the AMP achievable mean square error (MSE) can be predicted, and the region of system parameters where the AMP algorithm is sub-optimal can be determined. In addition, the phase transition phenomenon of our system can be analyzed, which guides the system design for massively connected networks.

Note that, different from \cite{LiuMassive12018}, where the performance analysis is performed based on the isotropic Raleigh channel model, the analysis in this paper is considered based on a more general channel model following the approach of \cite{JSDM2013}. Such a channel model has strong flexibility and reduces to particular channel models in multiple typical MIMO communication scenarios, such as the isotropic Raleigh channel model \cite{LiuSparse2018}, and the spatially correlated channel model \cite{MRAcoCH2020, CEMoG2015, CoCE2013}, etc. In particular, our performance analysis produces some novel results for both the isotropic Raleigh channel and spatially correlated channel case, and provides verifications for the analytical results in \cite{LiuSparse2018, LiuMassive12018}.
Our main contributions can be summarized as follow.
\begin{itemize}
\item
First, we consider a general grouping channel model for the joint AUD and CE scenario. We design a theoretical framework carefully tailored to our considered scenario based on the general idea of replica method. Based on that, we establish relations between the joint AUD and CE performance metrics and the system parameters: pilot length, transmit power, fraction of number of active users and number of BS antennas, under our considered channel model.

\item
Second, we analyze the isotropic channel scenario based on our theoretical framework. Concretely, we prove that the Bayes-optimal/AMP-achievable joint AUD and CE performance can be evaluated by a scalar valued function, which also provides a phase transition diagram. We further provide the analysis in the asymptotic MIMO regime and prove that the phase transition phenomenon disappears and the AMP can always achieve Bayes-optimal performance when the number of BS antennas is very large.

\item
Third, we analyze the performance of joint AUD and CE in the spatially correlated channel case. We prove that the performance of each user group can be separately evaluated by a particular free entropy function when the user groups are in mutually orthogonal subspaces, implying the per-group processing (PGP) will never bring the performance loss. In addition, we show that the spatially correlated channel is more likely to promote a phase transition compared with the isotropic channel, due to the small number of subspace dimensions and corresponding antenna gains, leading that
there will be a performance gap between AMP and Bayes-optimal performance. However, the pilot length required to step over the phase transition can be significantly reduced in that case.
\end{itemize}

%%The rest of this paper is organized as follow. Section II describes the system model and the grouping channel model for joint AUD and CE.
%Section III provides the derivations of our theoretical framework. Section IV analyzes and demonstrates the performance of joint AUD and CE in the isotropic channel case. Section V analyzes and demonstrates the performance of joint AUD and CE in the spatially correlated channel case. Section VI concludes this paper.
%
%
%\emph{Notation:} The identity matrix and the all-zero matrix of appropriate dimensions are denoted as $\boldsymbol I$ and $\boldsymbol 0$. For a matrix $\boldsymbol A$ of arbitrary size, $\boldsymbol A^H$ and $\boldsymbol A^T$ denote its conjugate transpose and transpose, respectively. The expectation operation is denoted as $\mathbb E(\cdot)$. Function ${\rm Tr}(\cdot)$ returns the trace of a matrix. The space of complex matrices of size $m\times n$ is denoted as $\mathbb C^{m\times n}$.

\begin{figure}[!t]
\centering
\includegraphics[width=3.5in]{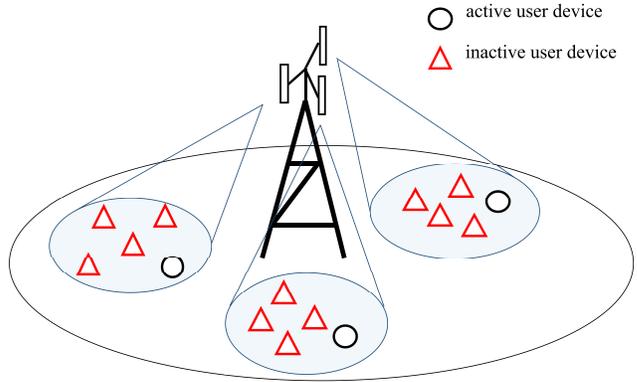}
\caption{Our model of the massive device communication network.}
\label{sys}
\end{figure}

\section{System Model}
\subsection{Massive Random Access Scenario}
We consider a uplink massive random access scenario in a single-cell cellular network with $N$ user devices. Each user is equipped with a single antenna, and the BS is equipped with $M$ antennas. The user traffic is assumed to be sporadic, i.e., most of user devices are idle at any given time.

Following the approach of \cite{JSDM2013}, in this paper, we consider the following channel model. We suppose that $N$ users are divided into $G$ groups based on the similarity of their covariance matrixes, where each group has $K_g$ users and the number of users is assumed same in each group for concise, i.e., $K_g = K = N/G$. Our model of the massive device communication network is shown in Fig. \ref{sys}. We denote $\boldsymbol h_{gk}\in\mathbb C^{M\times 1}$ as the channel response of $k$th user in the group $g$. We further assume that users in the same group have an identical channel probability density function (PDF) denoted as $Q_g(\boldsymbol h_{kg}) = \mathcal{CN}(\boldsymbol h_{kg}; \boldsymbol 0, \mathcal C_g)$ with covariance matrix $\mathcal C_g$. For the $k$th user in the group $g$, the user state, i.e., active or idle, is characterized by an activity indicator, denoted as $a_{gk}$, i.e., $a_{gk}= 1$ if it is active and $a_{gk}= 0$ otherwise.
We denote $\boldsymbol f_{kg}\in\mathbb C^{T\times 1}$ as the non-orthogonal pilot sequence that is assigned to the $k$th user in the group $g$ with independent and identically distributed (i.i.d.) components generated from a complex Gaussian distribution with zero mean and variance $1/T$, so that the pilot sequence has a unit norm\footnote{In the asymptotic regime that we consider, the pilot sequences have an unit power.}. Accordingly, the overall channel input-output is
\begin{equation}
\boldsymbol Y  = \boldsymbol F\boldsymbol S + \boldsymbol W = \sum\limits_{g,k} a_{gk}\boldsymbol f_{gk}\boldsymbol h_{gk}^{T}+ \boldsymbol W,  \label{channel}
\end{equation}
where $\boldsymbol F\in\mathbb C^{T\times N} \triangleq [\boldsymbol f_1,\dots, \boldsymbol f_{GK}]$ is the pilot matrix, $\boldsymbol S\in\mathbb C^{N\times M} \triangleq [\boldsymbol s_{11},\dots, \boldsymbol s_{GK}]^{T}$ with $\boldsymbol s_{gk}\in\mathbb C^{N\times 1}\triangleq a_{gk}\boldsymbol h_{gk}$, and $\boldsymbol W$ is the additive white Gaussian noise (AWGN) matrix with i.i.d. elements distributed as $\mathcal{CN}(0,\sigma_w^2)$. Each user device has a constant transmit power $P_t$, which is absorbed into the large-scale fading of the channel coefficients in (\ref{channel}) for concise.

We assume that the user devices are synchronized and each user decides whether or not to access the channel with a
probability $\rho$ in an i.i.d. manner \cite{LiuSparse2018, LiuMassive12018, MRAcoCH2020, HannakJoint2015}, and we further assume the fraction of active users $\rho$ and the PDF $Q_g(\boldsymbol h_{kg})$ of the $k$th user channel in the group $g$ are known for the BS\footnote{Note that the parameters learning strategy under mMTC scenario has been considered in the related work \cite{MRAcoCH2020, MAEC2020} based on the expectation-maximization method, which can be easily extended under our model. Since we focus on the theoretical analysis under our framework, the parameters learning issue is beyond the scope of this paper}. As a consequence, the corresponding PDF of the matrix $\boldsymbol S$ can be formulated as
\begin{equation}
p(\boldsymbol S) = \prod_{g=1}^{G}\prod_{k=1}^{K}[(1-\rho)\delta(\boldsymbol s_{gk})+\rho Q_g(\boldsymbol s_{gk})].\label{prs}
\end{equation}

%Before proceeding, we further define $\boldsymbol H\in\mathbb C^{N\times M} \triangleq [\boldsymbol h_1,\dots, \boldsymbol h_{GK}]^{T}$ as the collection of all user channels, and $\boldsymbol s_{gk}\in\mathbb C^{N\times 1}\triangleq a_{gk}\boldsymbol h_{gk}$ is denoted as user transmitted signal for the $k$th user in group $g$. After that, the general channel probabilistic model that includes all the cases is
%\begin{equation}
%p(\boldsymbol H) = \prod_{g=1}^{G}\prod_{k=1}^{K} Q_g(\boldsymbol h_{gk}),
%\end{equation}
%where function $Q_g(\cdot)$ is denoted as the specific channel probability distribution function (PDF) for group $g$.

%According to the definition of matrix $\boldsymbol S$, the corresponding PDF of the matrix can be formulated as
%
%where $\rho$ is the active ratio of users. Note that we are relying on the arguments to distinguish the different distributions in $g$ group, much as we will do with the following distributions.
%
Based on our model (\ref{channel}) and the prior distribution of the transmitted signals (\ref{prs})\footnote{Note that we assume joint AUD and CE is performed within one coherence time so that the channel coefficients remain unchanged.}, our task is to identify the active users, i.e., determine the active indicator $a_{gk}$ for each user, as well as estimate the corresponding channel coefficients $\boldsymbol h_{gk}$ for the active users.

\emph{Remark}: Note that our channel model assumption has a strong flexibility and reduces to particular channel models in multiple typical MIMO communication scenarios, such as the isotropic Rayleigh fading channels \cite{LiuSparse2018, LiuMassive12018}, where the channel coefficients of each user device are i.i.d., and the spatially correlated channels \cite{MRAcoCH2020, CEMoG2015, CoCE2013}, where different user groups are sufficiently well separated in the angular domain, and the channel covariance matrixes exhibit a low rank structure. We also note that although we consider the case where each group has the same number of users and the fraction of active users for each group is also the same, our following proposed theoretical framework can be easily extended to the case where the active ratio and user number are different for each group.

Based on our general channel model (\ref{prs}), in the following, we propose a novel theoretical framework by adopting replica method \cite{RepTan2002, 1985Entropy} to provide phase transition analysis, Bayes-optimality analysis of our scenario, and our theoretical framework can also provide the performance prediction of the AMP-based joint AUD and CE, which is the state-of-art algorithm under the Bayesian framework.

\section{Replica Analysis on Joint AUD and CE}
\label{a_re}

Our theoretical framework is based on statistical physics. In the statistical physics literature, the \emph{free entropy} of a system reflects the macro performance that characterizes some thermodynamic properties of the system \cite{1985Entropy}. Some related works in the signal processing literature have also shown that evaluating the fixed point of free entropy function provides the minimum MSE (MMSE) prediction for signal recovery and the achievable MSE for the AMP framework \cite{RepMMVzhu2018, RepMMVG2018}. For the massive connectivity scenario, evaluating the free entropy function of our system (\ref{channel}) provides an analytic tool to measure the performances metrics of joint AUD and CE.

Towards this end, we adopt the replica method in this paper to evaluate the free entropy function of our system. The replica method \cite{RepTan2002, 1985Entropy}, which is a classical tool in the statistical physics literature, provides an efficiency way to reduce the calculation of a certain free entropy into an optimization problem over specific covariance matrices. Such a reduction is based on a set of typical replica assumptions that include the self-averaging property, the validity of ``replica trick'', the ability to exchange certain limits and the replica symmetry \cite{2012Probabilistic}. It has been shown in many related works that the replica method provides correct predictions on many typical applications, such as signal processing \cite{RepCS2012} and physical layer communications \cite{RepADChe2018, RepADCW2017}.

Our analysis based on the replica method framework \cite{2012Probabilistic} is considered under a certain \emph{asymptotic regime} with $K\rightarrow\infty$, $T\rightarrow\infty$ and a fixed ratio $T/K\rightarrow\alpha$. Next, we provide the background of the statistical physics literature and obtain the free entropy function with respect to the system parameters under our scenario via replica method.
    Although these system parameters can not be infinity in the practice, we still utilize the asymptotic regime to facilitate analysis, since the system parameters are typically large in the massive connectivity. Authors in \cite{LiuMassive12018} have shown that the recovery performances in the practical settings match the theoretical results in the large system limit.

%In this section, we provide a theoretical framework based on replica method to find the Bayes-optimal solution, predict the AMP achievable performance and to characterize the phase transition under specific system configuration.
\setcounter{TempEqCnt}{\value{equation}}
\setcounter{equation}{6}
\begin{figure*}[tbp]
\begin{equation}
\begin{split}
\Phi(\mathcal E) = &-{\rm Tr}\left(\left(\boldsymbol\Delta+\frac{\mathcal EG}{\alpha}\right)^{-1}\left(\sum\limits_g\rho\mathcal C_g-\mathcal EG\right)\right)-\alpha M-\alpha \log\left|\boldsymbol\Delta+\frac{\mathcal EG}{\alpha}\right|\\
&+\sum\limits_g\int{\rm d}\boldsymbol s_g[(1-\rho)\delta(\boldsymbol s_g)+\rho Q_g(\boldsymbol s_g)]\int{\rm D}\boldsymbol z\log\left\{\int{\rm d}\boldsymbol x_g[(1-\rho)\delta(\boldsymbol x_g)+\rho Q_g(\boldsymbol x_g)]\right.\\
&\left.\exp\left\{-\boldsymbol x_g^H(\boldsymbol\Delta+\frac{\mathcal EG}{\alpha})^{-1}\boldsymbol x_g+2\mathfrak R\left(\boldsymbol x_g^H[(\boldsymbol\Delta+\frac{\mathcal EG}{\alpha})^{-1}\boldsymbol s_g+(\boldsymbol\Delta+\frac{\mathcal EG}{\alpha})^{-\frac{1}{2}}\boldsymbol z]\right)\right\}\right\}.\label{general_FE}
\end{split}
\end{equation}
\hrulefill
\end{figure*}
\setcounter{equation}{\value{TempEqCnt}}
\subsection{Free Entropy Function Derivation via Replica Method}
%\subsection{Replica Analysis}
%The free entropy function over our massive random access scenario can be derived via the replica method with some typical replica assumptions in the large system limits []-[].
Under the statistical physics literature, a probabilistic inference approach to reconstruct the transmitted signal $\boldsymbol S$ aims to sample $\boldsymbol X$ from the following posterior distribution based on (\ref{channel}), given by
\begin{equation}
p(\boldsymbol X|\boldsymbol Y) = \frac{p(\boldsymbol X)}{Z}\prod_{t=1}^{T}\frac{\exp\left(-(\boldsymbol y_t-\boldsymbol f_t\boldsymbol X)\boldsymbol\Delta^{-1}(\boldsymbol y_t-\boldsymbol f_t\boldsymbol X)^{H}\right)}{\pi^M|\boldsymbol\Delta|},\label{posterior}
\end{equation}
where $\boldsymbol\Delta\triangleq \sigma_w^2\boldsymbol I$, $Z$ is the partition function of this distribution, and $p(\boldsymbol X) =\prod_{g,k}[(1-\rho)\delta(\boldsymbol x_{gk})+\rho Q_g(\boldsymbol x_{gk})]$, which can be regarded as the prior distribution of $\boldsymbol X$. Under the Bayes-optimal assumption, the prior distribution over $\boldsymbol X$ matches the PDF of transmitted signal $p(\boldsymbol S)$. As a consequence, $Z$ can be formulated as
\begin{align}
Z = \int {\rm d} \boldsymbol X&\prod_{t=1}^{T}\frac{1}{\pi^M|\boldsymbol\Delta|}e^{-(\boldsymbol y_t-\boldsymbol f_t\boldsymbol X)\boldsymbol\Delta^{-1}(\boldsymbol y_t-\boldsymbol f_t\boldsymbol X)^{H}}\nonumber\\
&\times\prod_{g=1}^{G}\prod_{k=1}^{K}[(1-\rho)\delta(\boldsymbol x_{gk})+\rho Q_g(\boldsymbol x_{gk})],\label{patition}
\end{align}
where $\boldsymbol f_t$ is the $t$th row of the pilot matrix $\boldsymbol F$. We note that above integration is performed over each element of matrix $\boldsymbol X$ and such a integral expression will be used throughout this paper. According to the statistical physics literature, the free entropy is defined as $\Phi \triangleq \log Z$. Utilizing the self-averaging assumption \cite{RepTan2002, 1985Entropy}: $\lim\limits_{N\to \infty} \Pr\left[\left|\frac{Z}{N}-\frac{\mathbb E(Z)}{N}\right|\ge\theta\right] = 0$, for any tolerance $\theta>0$, the free entropy can be reformulated as $\Phi = \frac{1}{N}\mathbb E_{\boldsymbol F,\boldsymbol S,\boldsymbol W}(\log Z)$. Hence, to determine the free entropy function, one requires to compute the average of the log-partition function $\mathbb E_{\boldsymbol F,\boldsymbol S,\boldsymbol W}(\log Z)$. Based on the standard procedure of the replica method \cite{2012Probabilistic}, we can alternatively determine the average log-partition function by the following relation
$
\mathbb E_{\boldsymbol F,\boldsymbol S,\boldsymbol W}(\log Z) = \lim\limits_{n\to 0}\frac{1}{n}\log\left(\mathbb E_{\boldsymbol F,\boldsymbol S,\boldsymbol W}Z^n\right)
$.
According to the \emph{replica trick} in \cite{2012Probabilistic, RepTan2002, 1985Entropy}, the quantity $\mathbb E_{\boldsymbol F,\boldsymbol S,\boldsymbol W}Z^n$ is carried out as if $n$ were an integer, and we take the fact that $n$ is a real number into consideration after obtaining a manageable enough expression. As a consequence, the free entropy can be formulated as
\begin{equation}
\Phi  = \lim\limits_{N\to \infty}\lim\limits_{n\to 0}\frac{1}{Nn}\log\left(\mathbb E_{\boldsymbol F,\boldsymbol S,\boldsymbol W}Z^n\right).\label{re_trick}
\end{equation}
Based on this, we have the following theorem.
\begin{theorem}
Define the reconstruction MSE matrix over the system (\ref{channel}) as $\mathcal E\triangleq \frac{1}{G}\sum_g \mathcal E_g$, where \begin{equation}\mathcal E_g \triangleq \frac{1}{K}\sum\limits_{k = 1}^K\left(\hat{\boldsymbol x}_{gk}(\boldsymbol Y)-\boldsymbol s_{gk}\right)\left(\hat{\boldsymbol x}_{gk}(\boldsymbol Y)-\boldsymbol s_{gk}\right)^H,\end{equation} where $[\hat{\boldsymbol x}_{11}(\boldsymbol Y),\dots, \hat{\boldsymbol x}_{GK}(\boldsymbol Y)] = \hat{\boldsymbol X}$ is the sample of the posterior distribution $p(\boldsymbol X|\boldsymbol Y)$ with given $\boldsymbol Y$ defined in (\ref{posterior}). Under our certain asymptotic regime defined in Sec. \ref{a_re}, and using the standard assumptions of the replica method, i.e., the self-averaging assumption, the replica-symmetry assumption, and the replica trick in (\ref{re_trick}), the free entropy (\ref{re_trick}) can be calculated by optimizing a function with respect to $\mathcal E$, which is given in the equation (\ref{general_FE}). We note that the term ${\rm D}\boldsymbol z$ is a complex Gaussian integration measure.

\end{theorem}
\begin{proof}
Please see the appendix \ref{A}.
\end{proof}
We call the function (\ref{general_FE}) as \emph{free entropy function}\footnote{Since the free entropy function is closely related to the free entropy, we also use $\Phi$ to represent it. }.
According to \cite{RepMMVzhu2018, RepMMVG2018}, optimizing  $\mathcal E$ that maximizes the free entropy function (\ref{general_FE}) corresponds to the minimum MSE (MMSE) of specific choices of system parameters, i.e., $\rho$, $\alpha$, $\mathcal C_g$, $\forall g$, in the Bayes-optimal condition. Setting the first order derivative of $\Phi(\mathcal E)$ with respect to the matrix $\mathcal E$ into zero, we get the following equation.
\setcounter{equation}{7}
\begin{align}
\mathcal E = \frac{1}{G}\sum\limits_{g}&\mathbb E_{\boldsymbol s_g, \boldsymbol z}\left[\left(\eta_g\left(\boldsymbol s_g+(\boldsymbol\Delta+\frac{\mathcal EG}{\alpha})^{\frac{1}{2}}\boldsymbol z\right)-\boldsymbol s_g\right.\right)\nonumber\\
&\times\left(\left.\eta_g\left(\boldsymbol s_g+(\boldsymbol\Delta+\frac{\mathcal EG}{\alpha})^{\frac{1}{2}}\boldsymbol z\right)-\boldsymbol s_g\right)^H\right],\label{SE}
\end{align}
where $\boldsymbol z$ obeys $\mathcal{CN}(\boldsymbol z; \boldsymbol 0, \boldsymbol I)$ and $\eta_g(\cdot)$ is the Bayes-optimal denoiser of the noisy measurement $\hat{\boldsymbol s}_g\triangleq \boldsymbol s_g+(\boldsymbol\Delta+\frac{\mathcal EG}{\alpha})^{\frac{1}{2}}\boldsymbol z$ with the latent signal $\boldsymbol s_g$ distributed as $(1-\rho)\delta(\boldsymbol s_g)+\rho Q_g(\boldsymbol s_g)$. We note that $\eta_g(\cdot)$ is the corresponding MMSE denoiser of noisy measurement $\hat{\boldsymbol s}_g$.
%\begin{equation}
%\eta_g(\hat{\boldsymbol s}_g) = \frac{\left(\mathcal C_g^{-1}+\left(\boldsymbol\Delta+\frac{\mathcal EG}{\alpha}\right)^{-1}\right)^{-1}\left(\boldsymbol\Delta+\frac{\mathcal EG}{\alpha}\right)^{-1}\hat{\boldsymbol s}_g}{1+\frac{1-\rho}{\rho}
%\frac{|\boldsymbol\Delta+\frac{\mathcal EG}{\alpha}+\mathcal C_g|}{|\boldsymbol\Delta+\frac{\mathcal EG}{\alpha}|}\exp\left\{-\hat{\boldsymbol s}_g^H\left(\left(\boldsymbol\Delta+\frac{\mathcal EG}{\alpha}\right)^{-1}-\left(\boldsymbol\Delta+\frac{\mathcal EG}{\alpha}+\mathcal C_g\right)^{-1}\right)\hat{\boldsymbol s}_g\right\}}.
%\end{equation}

We can also observe from (\ref{SE}) that there is a close match between the fixed point of the AMP state evolution function in \cite{BayatiDynamics2011} and the stationary point of the free energy function (\ref{general_FE}). In addition, according to \cite{RepMMVzhu2018, RepMMVG2018}, the AMP algorithm will be blocked by a particular local maximum point of (\ref{general_FE}). This means that seeking the stationary points of the free entropy function (\ref{general_FE}) also provides accurate prediction of the performance of the AMP algorithm. We remark that the recovery MSE of AMP is sometimes sub-optimal since only the global maximum point that maximizes the free entropy function (\ref{general_FE}) corresponds to the MMSE, and we call it AMP-achievable MSE.

In the following, we establish connections between the fixed point of the AMP state evolution function predicted by the free entropy function (\ref{general_FE}), and the performance metrics of joint AUD and CE. Typically, we adopt the likelihood ratio test (LRT) detection for AUD, and adopt MMSE criterion for CE.  We note that the author in \cite{OnRangan} claims that, in the case of large i.i.d. zero-mean Gaussian sensing matrix, the AMP methods exhibit fast convergence. Since in the massive access literature, the system is large enough and the i.i.d. zero-mean Gaussian sensing matrix is adopted, guaranteeing the convergence of the AMP algorithm in the following sections.

\setcounter{TempEqCnt}{\value{equation}}
\setcounter{equation}{13}
\begin{figure*}[tbp]
\begin{equation}
\begin{split}
\Phi(\tau) = &-\alpha M\left(\frac{\sigma_w^2}{\frac{1}{\alpha}\tau+\sigma_w^2}+\log(\sigma_w^2+\frac{1}{\alpha}\tau)\right)
+M\sum\limits_g\frac{(1-\rho)\sigma_g^2}{\sigma_g^2+\sigma_w^2+\frac{1}{\alpha}\tau}\\
&+\sum\limits_{g}\int{\rm D}\boldsymbol z\rho\log\left[(1-\rho)\exp\left\{-\frac{||\boldsymbol z||^2\sigma_g^2}{\sigma_w^2+\frac{1}{\alpha}\tau}\right\}+\rho\left(\frac{\frac{1}{\alpha}\tau+\sigma_w^2}{\frac{1}{\alpha}\tau+\sigma_w^2+\sigma_g^2}\right)^M\right]\\
&+\sum\limits_{g}\int{\rm D}\boldsymbol z(1-\rho)\log\left[(1-\rho)\exp\left\{-\frac{||\boldsymbol z||^2\sigma_g^2}{\sigma_g^2+\sigma_w^2+\frac{1}{\alpha}\tau}\right\}+\rho\left(\frac{\frac{1}{\alpha}\tau+\sigma_w^2}{\frac{1}{\alpha}\tau+\sigma_w^2+\sigma_g^2}\right)^M\right].\label{FE_tau}
\end{split}
\end{equation}
\hrulefill
\end{figure*}
\setcounter{equation}{\value{TempEqCnt}}

\subsection{Prediction of AMP-based Joint AUD and CE}
\label{pr_audce}

%Accordingly, we analyze the detection performance in terms of missed detection and false alarm probabilities, and the channel estimation error is measured by MSE.

%We have clarified that finding the local maxima points of the free entropy function can provide the accurate prediction of the state evolution fix point of AMP framework. In this section, we establish connections between the fix point of state evolution and the performance metrics of AUD and CE, indicating that the free entropy function can provide accurate predictions of joint AUD and CE performances.

After executing the AMP algorithm on the transmitted signal $\boldsymbol Y$ in (\ref{channel}), the overall estimation problem is decoupled as a sequence of vector-valued estimation problems. For the users in the group $g$, the decoupled signal model is $\hat{\boldsymbol s}_g\triangleq \boldsymbol s_g+\boldsymbol \Sigma^{\frac{1}{2}}\boldsymbol z$, where we define $\boldsymbol \Sigma\triangleq \boldsymbol\Delta+\frac{\mathcal E^{\star}G}{\alpha}$ as the equivalent noise covariance matrix. The $\mathcal E^{\star}$ denotes the fixed point of AMP state evolution function that fulfills equation (\ref{SE}), so that it is also a stationary point of the free entropy function (\ref{general_FE}). Note that we have dropped the subscript about the user indicator $k$, much as the following expressions, since all the users within a group share a common probabilistic model.

As a consequence, for the user devices in the group $g$, the AUD rule based on log-likelihood ratio (LLR) with a threshold $l_g$ can be formulated as
\begin{align}
\text{LLR}(\hat{\boldsymbol s}_g) =& \log\left(\frac{p(\hat{\boldsymbol s}_g|a_g = 1)}{p(\hat{\boldsymbol s}_g|a_g = 0)}\right)\nonumber\\
 =&\hat{\boldsymbol s}_g^H\tilde{\boldsymbol\Sigma}\hat{\boldsymbol s}_g+\log\frac{\left|\boldsymbol\Sigma\right|}{\left|\mathcal C_g+\boldsymbol\Sigma\right|}>l_g,\label{llr}
\end{align}
where we define $\tilde{\boldsymbol\Sigma}\triangleq\boldsymbol\Sigma^{-1}-\left(\mathcal C_g+\boldsymbol\Sigma\right)^{-1}$. Accordingly, the missed detection and false alarm probabilities $P_{M}$ and $P_{F}$ can be obtained via
\begin{align}
P_{F} = \Pr\left\{\hat{\boldsymbol s}_g^H\tilde{\boldsymbol\Sigma}\hat{\boldsymbol s}_g>l_g'|a_g = 0\right\},\label{pf}\\
P_{M} = \Pr\left\{\hat{\boldsymbol s}_g^H\tilde{\boldsymbol\Sigma}\hat{\boldsymbol s}_g<l_g'|a_g = 1\right\}\label{pm},
\end{align}
where $l_g' \triangleq l_g-\log\left|\boldsymbol\Sigma\right|+\log\left|\mathcal C_g+\boldsymbol\Sigma\right|$.

We then consider the performance analysis for CE. Since the CE for the active users are performed after the AUD, we thus consider the mean of conditional posterior distribution $p(\boldsymbol s_g|\hat{\boldsymbol s}_g, a_g = 1)$ as the channel estimator for active user devices. Based on the definition of noisy measurement $\hat{\boldsymbol s}_g$,
after performing the AMP algorithm, the CE for the active user devices in the group $g$ is
\begin{equation}
\hat{\boldsymbol h}_g = \left(\boldsymbol\Sigma^{-1}+\mathcal C_g^{-1}\right)^{-1}\boldsymbol \Sigma^{-1}\hat{\boldsymbol s}_g.\label{chan_est}
\end{equation}
%we then have
%\[\begin{split}p(\boldsymbol s_g|&\hat{\boldsymbol s}_g, a_g = 1) \\
%= &\mathcal{CN}\left(\boldsymbol s_g; \left(\boldsymbol\Sigma^{-1}+\mathcal C_g^{-1}\right)^{-1}\boldsymbol \Sigma^{-1}\hat{\boldsymbol s}_g, \left(\boldsymbol\Sigma^{-1}+\mathcal C_g^{-1}\right)^{-1}\right).\end{split}\]Therefore,

For the ease of analysis, we only consider the CE error in the case where the AUD is executed perfectly. Under such a circumstance, the CE error matrix can be formulated as
\begin{align}
\mathbb E\left[\left(\boldsymbol s_g-\hat{\boldsymbol h}_g\right)\left(\boldsymbol s_g-\hat{\boldsymbol h}_g\right)^H\right]
 = \left(\mathcal C_g^{-1}+\boldsymbol\Sigma^{-1}\right)^{-1},\label{cha_err_ma}
\end{align}
where the expectation is taken over $p(\boldsymbol s_g,\hat{\boldsymbol s}_g| a_g = 1)$.
As a consequence, the performance of joint AUD and CE can be analyzed after obtaining the stationary point $\mathcal E^{\star}$ of the free entropy (\ref{general_FE}). We note that the equation (\ref{general_FE}) and the corresponding performance of joint AUD and CE in this section is based on the general channel model (\ref{prs}). In the following sections, we will show our theoretical framework can reduce to two typical scenarios in the massive random access literature, i.e., the isotropic channel scenario and the spatially correlated channel scenario. Concretely, in Sec \ref{iso}, we provide the performance analysis of the isotropic channel case. While in the Sec. \ref{spc}, we provide the performance analysis of the spatially correlated channel case.

% and we will find that studying the free entropy function (\ref{general_FE}).

% we will observe a phase transition phenomenon and the region of system parameters where the AMP algorithm is suboptimal persists, but becomes smaller and eventually disappears as the variation of system parameters, so that the Bayes-optimality is achieved.

%using saddle point method $\mathcal F(\boldsymbol Q_g,\hat{\boldsymbol Q_g},\boldsymbol M_g, \hat{\boldsymbol M_g},\boldsymbol D_g, \hat{\boldsymbol D_g})$ can be expressed with respect to $\boldsymbol E\triangleq \{\boldsymbol E_{g}, g = 1,\dots,G\}$. The maximum point of the function $\mathcal F(\boldsymbol E)$ is the Bayes-optimal error matrix.

\section{Isotropic Channel}\label{iso}
In this section, we consider a typical scenario in wireless communications, where there are many small reflectors around the BS, and the Rayleigh fading MIMO channel with isotropic channel is assumed \cite{LiuSparse2018, LiuMassive12018}. In this case, the Rayleigh channel components of each user group are assumed to be i.i.d. complex Gaussian with zero means and unit variances across all the BS antennas, and function $Q_g(\boldsymbol h_g)$ in the prior distribution over the channel in the $g$th user group is $Q_g(\boldsymbol h_g) = \mathcal {CN}(\boldsymbol h_g;0,\sigma_g^2\boldsymbol I)$, where $\sigma_g^2$ is the product of the large scale coefficient and the user transmitted power. In the following, we will analyze the performance of joint AUD and CE under this case by first deriving the corresponding free entropy function as well as its expression under an asymptotic massive MIMO regime. Then, the performance analysis of joint AUD and CE is given.

\subsection{Free Entropy Function}
%We now derive the free entropy function under the isotropic Rayleigh channel scenario. As we have claimed, the performance metrics of joint AUD and CE can be predicted by acquiring the fix point of AMP state evolution, which can be derived by finding the extreme point of free entropy function (\ref{general_FE}).
Before deriving the concrete free entropy function, we note that it is proved in \cite{LiuMassive12018} that the equivalent noise covariance matrix in the AMP state evolution iterations always in the form of a diagonal matrix with identical diagonal entries in the isotropic Rayleigh channel case, so that the fixed point of the state evolution function (\ref{SE}) remains the same form. We therefore restrict the domain of $\mathcal E$ in the form of $\mathcal E  = \tau G^{-1}\mathbf I$, and the newly derived free entropy function with respect to $\tau$ is in the following theorem.

\begin{theorem}
With the assumption of the isotropic Rayleigh channel, i.e., $Q_g(\boldsymbol s_g) = \mathcal {CN}(\boldsymbol s_g;0,\sigma_g^2\mathbf I)$ and the matrix $\mathcal E$ is diagonal with identical entries, i.e., $\mathcal E = \tau G^{-1}\mathbf I$,
the free entropy function can be reformulated as a scalar valued function, given by (\ref{FE_tau}).

\end{theorem}

\begin{proof}
Please see the appendix \ref{T2}.
\end{proof}

Note that the MMSE under the isotropic channel case can be derived by calculating the global maximum point of the free entropy function (\ref{FE_tau}), and the fixed point of the AMP state evolution can be predicted by the largest $\tau$ that
associated with a local maximum of (\ref{FE_tau}). As we shall see in the following Sec. \ref{Num_1}, there exists a region of system parameters, i.e., pilot length $L$, transmit power $P_t$, fraction of number of active user devices $\rho$ and number of BS antennas $M$, where the MSE of AMP is blocked by a local maximum point, so that the AMP algorithm is suboptimal. However, the following proposition indicates that when considering an asymptotic MIMO regime with the number of BS antennas go to infinity, the region where the AMP algorithm is suboptimal vanishes.

%We then consider the asymptotic massive MIMO regime, where the number of BS antennas is assumed infinity. This leads a further simplification of the free entropy function (\ref{FE_tau}) and provide further insight of recovery performance under the scenario.

\begin{proposition}
In an asymptotic massive MIMO regime, i.e., $M\to \infty$, by neglecting some constants and the irrelevant factor, the free entropy function (\ref{FE_tau}) is reduced to
\setcounter{equation}{14}
\begin{align}
\Phi(\tau) = &-\frac{\alpha\sigma_w^2}{\frac{1}{\alpha}\tau+\sigma_w^2}-\alpha\log(\sigma_w^2+\frac{1}{\alpha}\tau)\nonumber\\
&-\rho\sum\limits_g\log(1+\frac{\sigma_g^2}{\sigma_w^2+\frac{1}{\alpha}\tau}),\label{FE_tau_limit}
\end{align}
which has and only has one local maximum point, and can be derived by calculating the equation
\begin{equation}
\tau = \rho\sum\limits_g\left(\sigma_g^{-2}+\left(\frac{1}{\alpha}\tau+\sigma_w^2\right)^{-1}\right)^{-1}.\label{t_opti}
\end{equation}
\end{proposition}
\begin{proof}
Please see the appendix \ref{P1}.
\end{proof}
Therefore, with an asymptotic number of BS antennas, since there is only one local maximum point, the AMP algorithm always achieves Bayes-optimal performance in the isotropic Rayleigh channel case, and the fixed point of AMP state evolution can be derived by solving (\ref{t_opti}).

\subsection{Prediction of AMP-based AUD and CE}
%After that, we begin to analyze the performance of AMP-based joint AUD and CE under the isotropic Rayleigh channel scenario based on the state evolution fix point predicted by the largest $\tau$ of function (\ref{FE_tau}) associated with the local maxima.

%After utilizing AMP framework, the equivalent signal model for user in group $g$ is given by $\hat{\boldsymbol s_g} = \boldsymbol s_g+(\boldsymbol\Delta+\frac{\tau}{\alpha}\mathbf I)^{-\frac{1}{2}}\boldsymbol z$. As a consequence, we obtain the conditional probability as
%\begin{align}
%&p(\hat{\boldsymbol s_g}|a_g = 1) = \mathcal{CN}\left(\hat{\boldsymbol s_g}; \boldsymbol 0, (\sigma_g^2+\sigma_w^2+\frac{1}{\alpha}\tau)\mathbf I\right),\\
%&p(\hat{\boldsymbol s_g}|a_g = 0) = \mathcal{CN}\left(\hat{\boldsymbol s_g}; \boldsymbol 0, (\sigma_w^2+\frac{1}{\alpha}\tau)\mathbf I\right).
%\end{align}
\begin{figure*}[t]
\begin{center}
\subfigure[$P_t = 18$dBm.]{
\begin{minipage}[t]{0.31\textwidth}
\includegraphics[width=2.4 in]{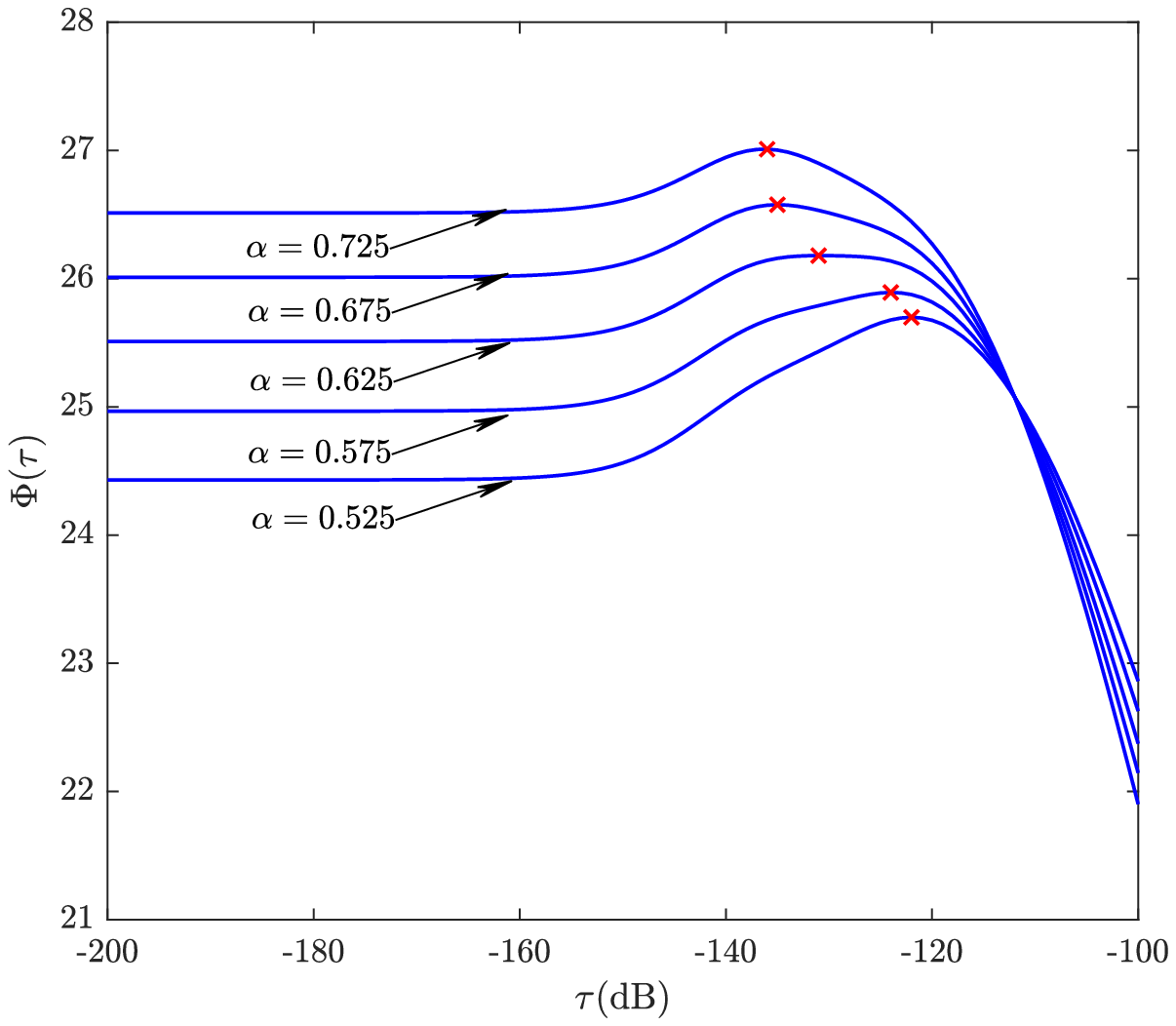}\label{18M2}
\end{minipage}
}
\hfill
\subfigure[$P_t = 23$dBm.]{
\begin{minipage}[t]{0.31\textwidth}
\includegraphics[width=2.4 in]{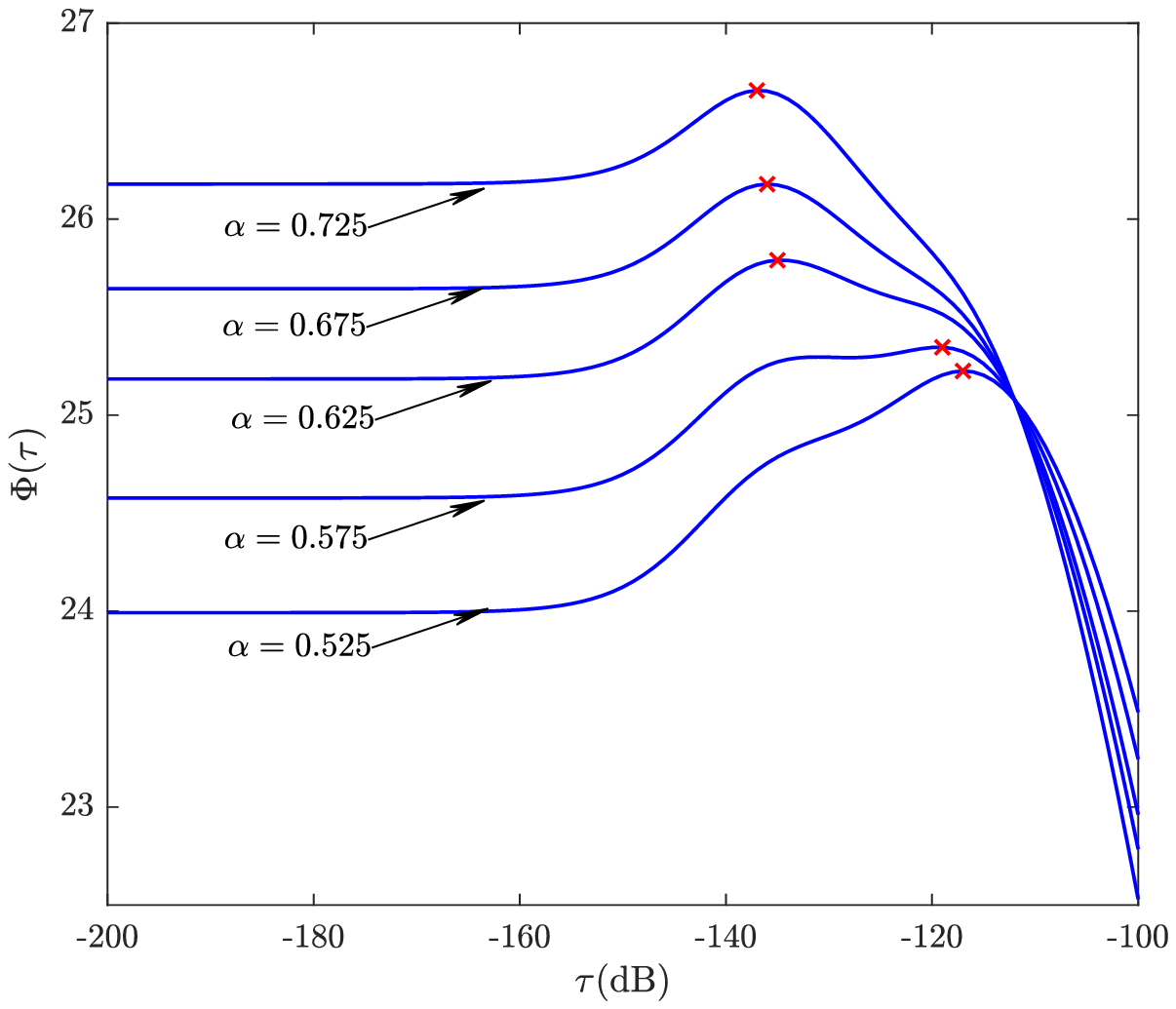}\label{23M2}
\end{minipage}
}
\subfigure[$P_t = 33$dBm.]{
\begin{minipage}[t]{0.31\textwidth}
\includegraphics[width=2.4 in]{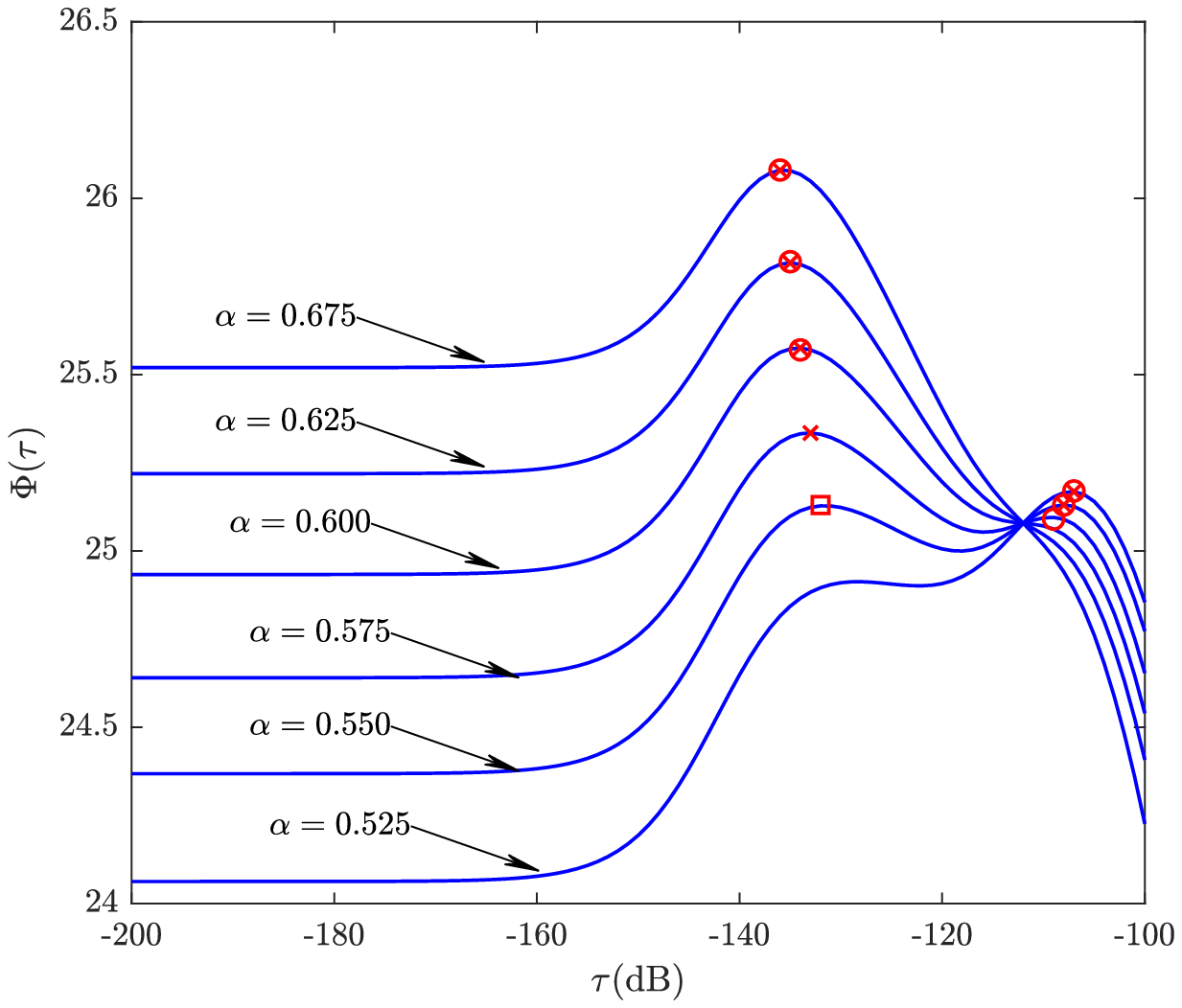}\label{33M2}
\end{minipage}
}
\caption{Free entropy as a function of MSE with $M = 2$ BS antennas under different settings of $\alpha$ and $P_t$. (The ``red cross'' denotes the MMSE point, the ``red circle'' denotes the AMP-achievable MSE point, the ``red square'' denotes the unachievable local maxima MSE point, and such marks will also be used in the following figures.)}
%\label{fig2}
\end{center}
\begin{center}
\subfigure[$M = 4$.]{
\begin{minipage}[t]{0.31\linewidth}
\centering
\includegraphics[width=2.4in]{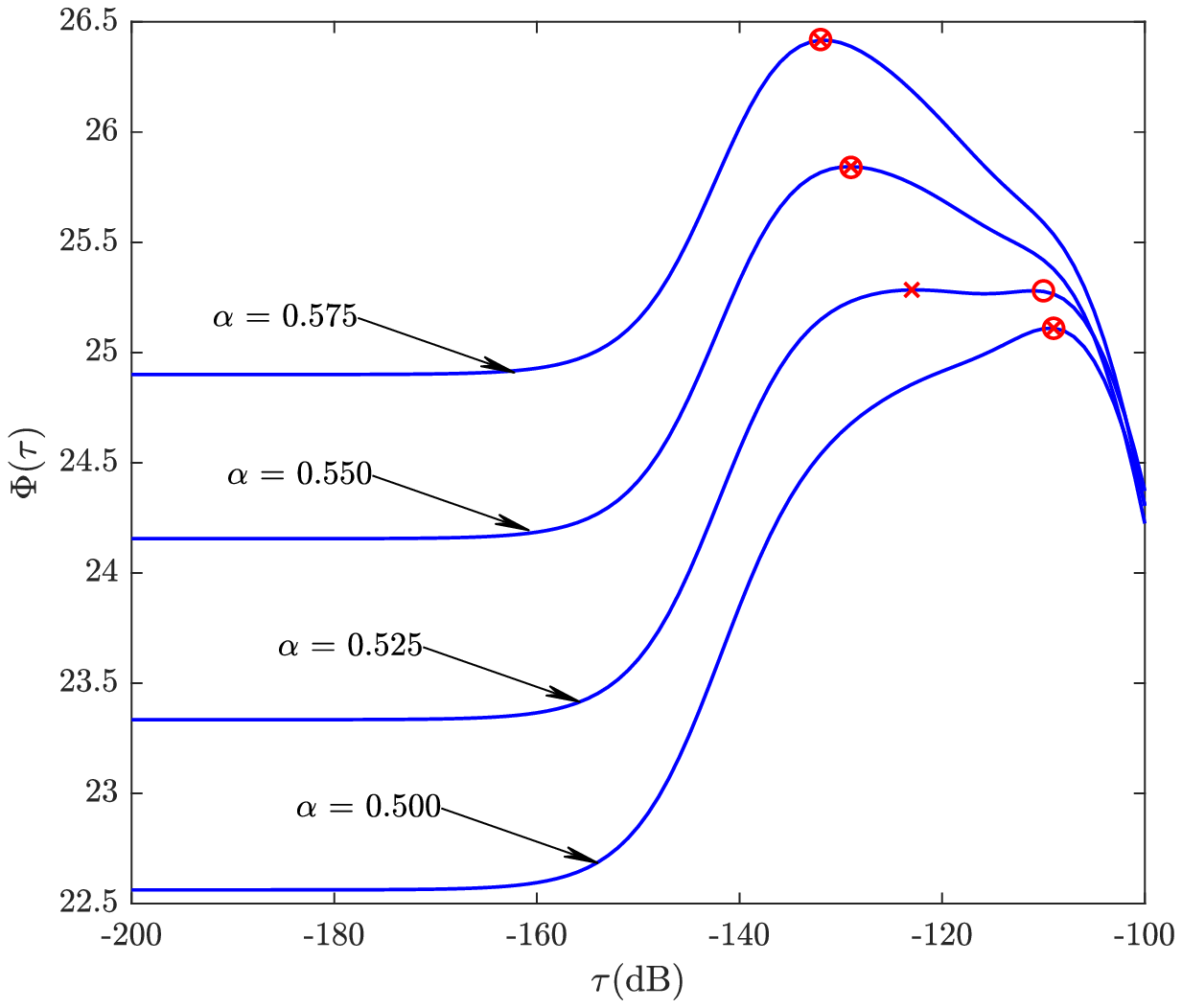}
%\caption{fig1}
\label{33M4}
\end{minipage}%
}
\subfigure[$M = 8$.]{
\begin{minipage}[t]{0.31\linewidth}
\centering
\includegraphics[width=2.4in]{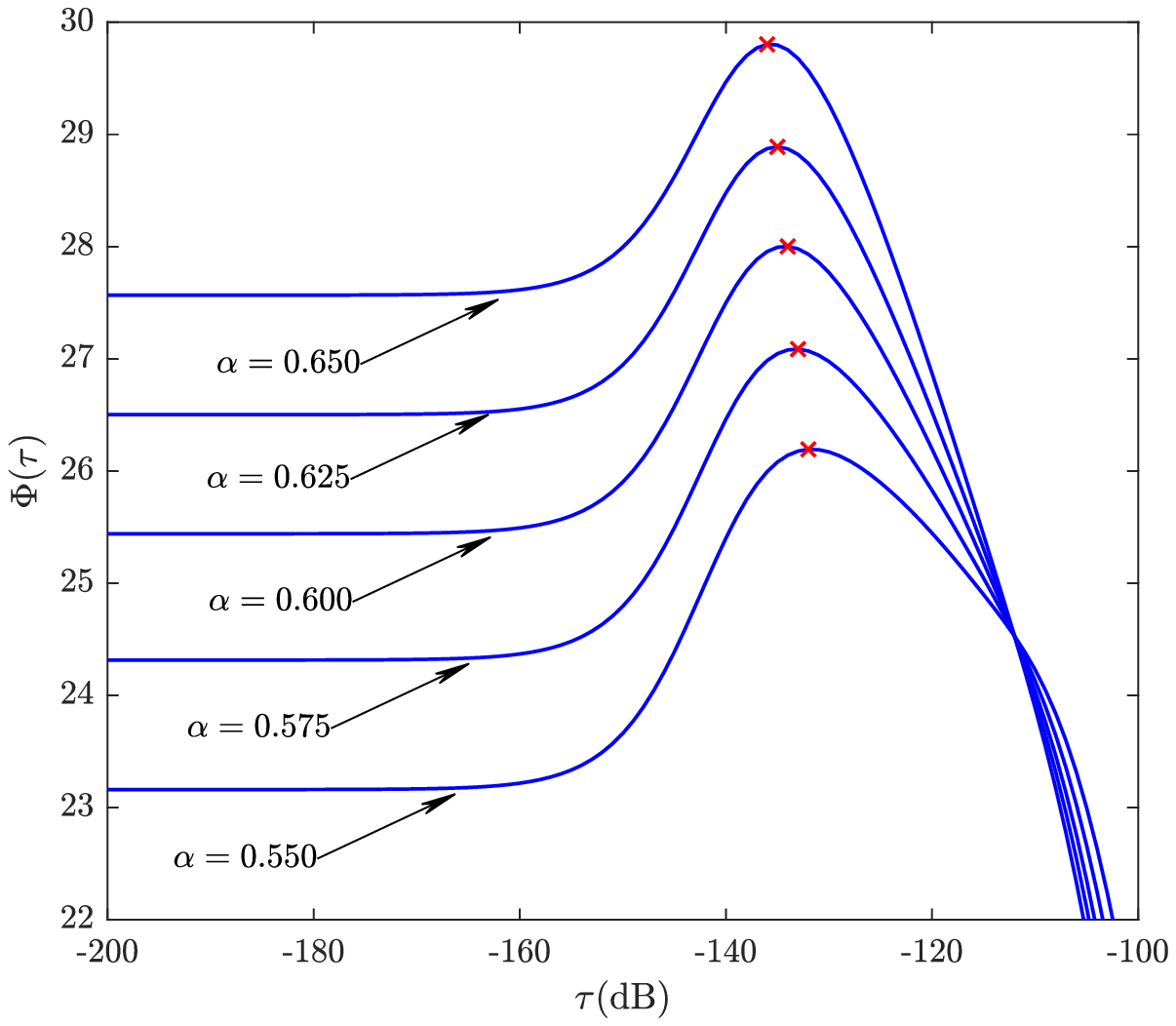}
%\caption{fig2}
\label{33M8}
\end{minipage}
}
\subfigure[$M\to \infty$.]{
\begin{minipage}[t]{0.31\linewidth}
\centering
\includegraphics[width=2.4in]{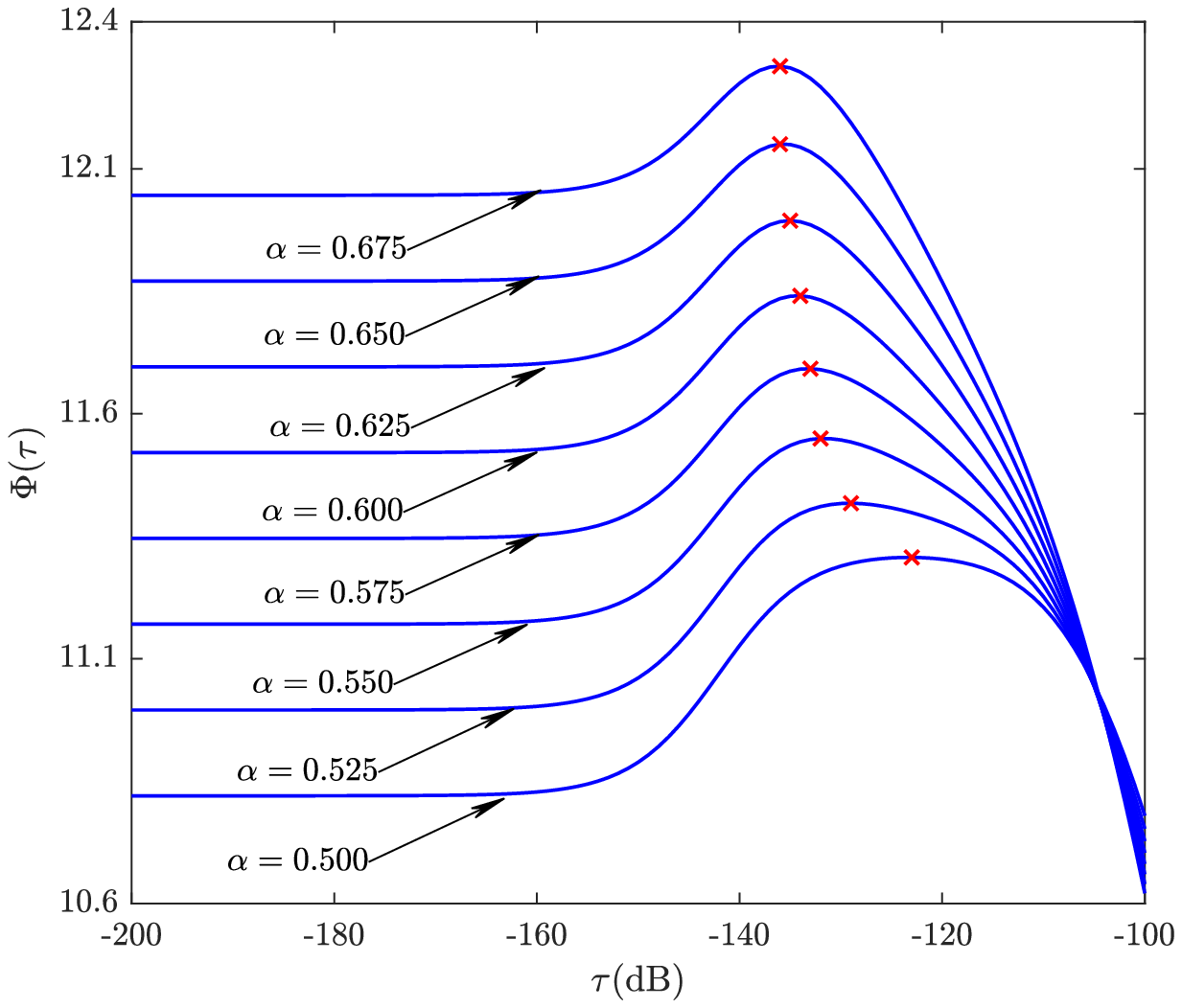}
%\caption{fig3}
\label{33Mlim}
\end{minipage}
}
\caption{Free entropy as a function of MSE with $P_t = 33$dBm transmit power under different settings of $\alpha$ and $M$.}
\end{center}
\end{figure*}
Define $\tau^{\star}$ as the largest stationary point of (\ref{FE_tau}) that associated with the fixed point of the AMP state evolution. Together with (\ref{llr}), the LLR in the isotropic channel case is
\begin{align}
&\text{LLR}(\hat{\boldsymbol s_g}) = M\log\frac{\sigma_w^2+\frac{1}{\alpha}}{\sigma_g^2+\sigma_w^2+\frac{1}{\alpha}\tau^{\star}}\nonumber\\
&+\hat{\boldsymbol s_g}^H\left((\sigma_w^2+\frac{1}{\alpha}\tau^{\star})^{-1}-(\sigma_g^2+\sigma_w^2+\frac{1}{\alpha}\tau^{\star})^{-1}\right)\hat{\boldsymbol s_g}.
\end{align}
Obviously, the corresponding detection sufficient statistic is $||\hat{\boldsymbol s_g}||^2$. As a consequence, we have
\begin{align}
P_D(M) &= \gamma\left(M, \left(\sigma_g^2+\sigma_w^2+\frac{1}{\alpha\tau^{\star}}\right)^{-1}l_g'\right)\Gamma^{-1}(M),\label{pm1}\\
P_{F}(M) &= \gamma\left(M, \left(\sigma_w^2+\frac{1}{\alpha\tau^{\star}}\right)^{-1}l_g'\right)\Gamma^{-1}(M).\label{pf1}
\end{align}
where $\gamma(M, \cdot)\Gamma(M)^{-1}$ is the CDF of the Chi-square distribution with $2M$ degrees of freedom, and
\begin{align}
l_g'\triangleq\sigma_g^{-2}&\left(\sigma_w^2+\frac{1}{\alpha\tau^{\star}}\right)\left(\sigma_g^2+\sigma_w^2+\frac{1}{\alpha\tau^{\star}}\right)\nonumber\\
&\left(l_g-M\log\frac{\sigma_w^2+\frac{1}{\alpha}\tau^{\star}}{\sigma_g^2+\sigma_w^2+\frac{1}{\alpha}\tau^{\star}}\right)\nonumber,
\end{align}
with arbitrary threshold $l_g$. Further, in the asymptotic massive MIMO regime, as $M$ go to infinity, we have $\lim\nolimits_{M\to \infty} P_{F}(M) = P_{M}(M) = 0$.

As for the CE performance, according to (\ref{chan_est}) and $\mathcal E = \tau G^{-1}\mathbf I$, after performing the AMP algorithm, the CE error for an active user in the group $g$ is
\begin{equation}
\mathbb E\left[\left(\boldsymbol s_g-\hat{\boldsymbol h}_g\right)\left(\boldsymbol s_g-\hat{\boldsymbol h}_g\right)^H\right]= \epsilon \mathbf I,
\end{equation}
where $\epsilon = \left(\left(\sigma_w^2+\frac{1}{\alpha}\tau^{\star}\right)^{-1}+\sigma_g^{-2}\right)^{-1}$.

Therefore, in the isotropic Rayleigh channel, perfect AUD and the Bayes-optimal channel estimation can be achieved as long as the number of antennas is large enough. We also note the above results validate and extend the analytical results in \cite{LiuMassive12018}. Besides, studying the free entropy function (\ref{FE_tau}) provides a phase transition diagram and an optimality analysis of the isotropic channel case with a finite number of BS antennas, as shown in the following section.

\subsection{Verification in the Isotropic Channel Scenario}
\label{Num_1}
In this section, we provide numerical examples to verify our results based on the replica method in the isotropic channel scenario. Specifically, we consider that each user accesses the channel with a probability $\rho = 0.1$ and we assume the served user devices have been divided into $5$ user groups and the distance $d_g$ between each user group is randomly distributed in the regime $[0.1, 1]$km. The path loss model of the wireless channel for each user group $g$ is given as ${\rm PL}_g = -128.1-36.7\log_{10}(d_g)$ in dB. Since, we assume no power adaption, we denote the transmit power for each user as $P_t$, and we have $\sigma_g^2 = P_t\times{\rm PL}_g$. We assume the bandwidth of the wireless channel are $1$MHz and the power spectral density of the AWGN at the BS is $-169$dBm/Hz. The
\subsubsection{Phase Transition of the Free Entropy Function}
First, we examine the phase transition of the free entropy function and demonstrating the optimal recovery MSE and the AMP-achievable MSE under our scenario, showing that for different settings of the system parameters, i.e., $P_t$ and $\alpha$, the MSE performance can be divided into some performance regions.

Specifically, we can notice that in the Fig. \ref{33M2}, there exists phase transitions and the MSE performance can be divided into several regions. In the first region, with $\alpha = 0.525$, the free entropy function has only one local maximum, indicating the optimal MSE performance, and the AMP-achievable MSE coincide. In the second region, with $\alpha = 0.550$, the second local maximum point with a lower function value than the first local maximum of the free entropy function appears, indicating the unachievable MSE point. In the third region, with $\alpha = 0.575$, the smaller local maximum point leads to a larger value of the free entropy function, which indicates the MMSE under such a choice of parameters. However, a larger local maximum point blocks the MSE performance of AMP, since it always converges to the local maximum point associated with the largest MSE. Hence, in this region, there exists a gap between the AMP-achievable MSE and the MMSE. Finally, in the region with $\alpha = 0.600\sim0.675$, one local maximum point disappears, and the AMP-achievable MSE and the MMSE coincide again. We can intuitively infer that there exists a hard threshold between the latter two regions, i.e, when the pilot length exceeds the threshold, the local maximum point associated with the largest MSE $\tau$ discards, and the AMP algorithm converges to the local maximum associated with a lower MSE, leading that there is a phase transition of the MSE performance of AMP.

We remark that in order to reduce the gap between the AMP-achievable MSE and the MMSE in the case that the AMP algorithm is sub-optimal, one can consider the idea of \emph{seeding matrix} introduced in \cite{2012Probabilistic} in the design of pilot matrix $\boldsymbol F$, which is beyond the scope of this manuscript.

In addition, Figs. \ref{18M2}-\ref{33M2} demonstrate the free entropy as a function of MSE with $M = 2$ BS antennas, and Figs. \ref{18M2}-\ref{33M2} demonstrate the free entropy function with $P_t = 33$dBm transmit power. We can notice that the region where AMP is suboptimal persists, but becomes smaller and eventually disappears as the number of BS antennas increases and transmit power decreases. Interestingly, we find that in the isotropic channel scenario, $M = 8$ antennas is enough to avoid phase transition in the case with the transmit power is lower than $P_t = 33$dBm.

\subsubsection{Performance Analysis of AUD}\label{Per_aud1}
\begin{figure}[ht]
\centering
\includegraphics[width=3.5in]{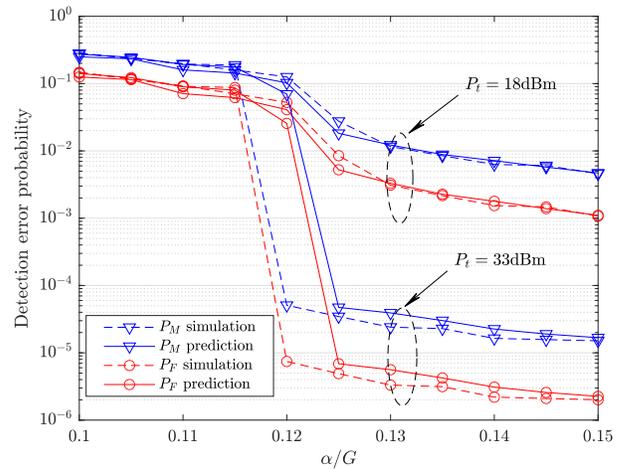}
\caption{Detection performance prediction of AMP with different settings of pilot length.}
\label{dtpre1}
\end{figure}

%\begin{figure}
%%\begin{minipage}[t]{0.5\linewidth}
%\centering
%\includegraphics[width=3.5in]{fig/det_perfM2.eps}
%\caption{Probabilities of missed detection and false alarm as functions of pilot length under different settings of transmit powers.}
%\label{dtperpt1}
%\end{figure}
%%\end{minipage}\quad%
%%\begin{minipage}[t]{0.5\linewidth}
%\begin{figure}
%\centering
%\includegraphics[width=3.5in]{fig/det_perfpt.eps}
%\caption{Probabilities of missed detection and false alarm as functions of pilot length under different settings of BS antennas.}
%\label{dtperm1}
%%\end{minipage}
%\end{figure}
We then provide numerical results to analyze the missed detection and false alarm probabilities in the isotropic channel scenario. The number of user devices is set as $N = 10000$. Fig. \ref{dtpre1} demonstrates the prediction of AMP-based AUD by equation (\ref{pf1}) and (\ref{pm1}) versus pilot length, with $\tau^{\star}$ predicted by the free entropy function. The simulation curves are depicted by the empirical AMP algorithm \cite{ChenSparse2018, RepMMVG2018}. Different from the standard AMP algorithm, the empirical AMP algorithm adopts an empirical state evolution function in the iterations, substituting the standard state evolution function in the AMP algorithm, in order to avoid the complex expectation calculations. The numerical results shows that our predicted performances are consistent with that of the empirical AMP algorithm for most of the settings. The empirical AMP and the proposed theoretical prediction have the similar tendency and there exists phase transition phenomenons in both of them. We notice that when phase transition occurs, the required length of pilot sequences for our prediction is slightly shorter than that for the empirical algorithm. Since the prediction error is about $0.04$ over the length of pilot sequence, we think the theoretical results can provide a prediction for the detection performance in the practice.  We also note that our predicted performance corresponding to the AMP achievable MSE provides a bound of the performance of the empirical AMP algorithm. The performance loss of the empirical AMP algorithm is because that the state evolution function that fulfills the fixed point condition (\ref{SE}) of the free entropy function is substituted by an empirical one.

%
%We note that the Besides, it is observed that the numerical results from empirical AMP algorithm match the predictions of probabilities of missed detection and false alarm very well, providing a guideline of the pilot length design based on specific requirement of detection metrics.

%Figs. \ref{dtperpt1}-\ref{dtperm1} investigate the detection performance under different settings of pilot length, transmit power and the number of BS antennas. We notice that with a small number of antennas, increasing the number of pilots has a slight impact on the detection performance both in the regions before and after the phase transition.

%In addition, it is observed that both decreasing the transmit power and increasing the antennas number will reduce the phase transition range. Interestingly, the phase transition caused by power increment benefits the detection performance.

\setcounter{TempEqCnt}{\value{equation}}
\setcounter{equation}{20}
\begin{figure*}[tbp]
\begin{align}
&\Phi_g(\mathcal E_g) = -\alpha{\rm Tr}\left(\left(\boldsymbol\Delta+\frac{\mathcal E_g}{\alpha}\right)^{-1}\left(\frac{1}{\alpha}\rho\mathcal C_g+\boldsymbol\Delta\right)\right)-\alpha\log\left|\boldsymbol\Delta+\frac{\mathcal E_g}{\alpha}\right|+\int{\rm d}\boldsymbol s_g[(1-\rho)\delta(\boldsymbol s_g)+\rho Q_g(\boldsymbol s_g)]\int{\rm D}\boldsymbol z\nonumber\\
&\log\left(\int{\rm d}\boldsymbol x_g[(1-\rho)\delta(\boldsymbol x_g)+\rho Q_g(\boldsymbol x_g)]\right.\exp\left(-\boldsymbol x_g^H(\boldsymbol\Delta+\frac{\mathcal E_g}{\alpha})^{-1}\boldsymbol x_g+2\mathfrak R\left(\boldsymbol x_g^H[(\boldsymbol\Delta+\frac{\mathcal E_g}{\alpha})^{-1}\boldsymbol s_g+(\boldsymbol\Delta+\frac{\mathcal E_g}{\alpha})^{-\frac{1}{2}}\boldsymbol z]\right)\right).\label{d_free}
\end{align}
\begin{align}
{\Phi}_g(\boldsymbol\Xi_g) = &-\alpha\sum\limits_{m = 1}^{r_g}\left(\frac{\sigma_w^2}{\frac{1}{\alpha}\xi_{g,m}+\sigma_w^2}+\log\left(\frac{1}{\alpha}\xi_{g,m}+\sigma_w^2\right)\right)+\sum\limits_{m = 1}^{r_g}\frac{(1-\rho)\lambda_{g,m}}{\lambda_{g,m}+\sigma_w^2+\frac{1}{\alpha}\xi_{g,m}}\nonumber\\
&+\int{\rm D}\boldsymbol z\rho\log\left[(1-\rho)\prod_{m = 1}^{r_g}\exp\left\{-\frac{|z_m|^2\lambda_{g,m}}{\sigma_w^2+\frac{1}{\alpha}\xi_{g,m}}\right\}+\rho\prod_{m = 1}^{r_g}\frac{\frac{1}{\alpha}\xi_{g,m}+\sigma_w^2}{\frac{1}{\alpha}\xi_{g,m}+\sigma_w^2+\lambda_{g,m}}\right]\nonumber\\
&+\int{\rm D}\boldsymbol z(1-\rho)\log\left[(1-\rho)\prod_{m = 1}^{r_g}\exp\left\{-\frac{|z_m|^2\lambda_{g,m}}{\lambda_{g,m}+\sigma_w^2+\frac{1}{\alpha}\xi_{g,m}}\right\}+\rho\prod_{m = 1}^{r_g}\left(\frac{\frac{1}{\alpha}\xi_{g,m}+\sigma_w^2}{\frac{1}{\alpha}\xi_{g,m}+\sigma_w^2+\lambda_{g,m}^2}\right)\right].\label{fe_si}
\end{align}
\hrulefill
\end{figure*}
\setcounter{equation}{\value{TempEqCnt}}

\subsubsection{Performance Analysis of CE}
%We have already noticed from sec. \ref{Per_aud1} that with $M = 8$ antennas, the phase transition of the detection performance discards, and detection error probabilities approach to zero even with very a small number of pilots. Hence, it is reasonable that we consider only the CE error in the case that the AUD is executed perfectly. The corresponding numerical results can be seen in the following.

\begin{figure}[ht]
\centering
\includegraphics[width=3.5in]{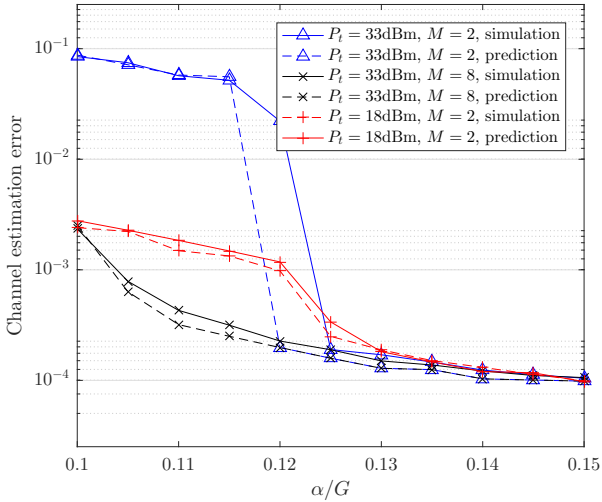}
\caption{Channel estimation performance prediction of AMP with different settings of pilot length.}
\label{cepre1}
\end{figure}

%\begin{figure}
%\begin{minipage}[t]{0.5\linewidth}
%\centering
%\includegraphics[width=3in]{fig/ce_perfM2.eps}
%\caption{Channel estimation error as functions of pilot length under different settings of transmit powers.}
%\label{cept1}
%\end{minipage}\quad
%\begin{minipage}[t]{0.5\linewidth}
%\centering
%\includegraphics[width=3in]{fig/ce_perfpt.eps}
%\caption{Channel estimation error as functions of pilot length under different settings of BS antennas.}
%\label{cem1}
%\end{minipage}
%\end{figure}
The results of the CE error prediction is demonstrated in Fig. \ref{cepre1}. We can observe that the CE error prediction performance is consistent compared with the detection error prediction. Our numerical results show that when the phase transition occurs, the performances of the joint AUD and CE are highly improved. The predicted performance corresponding to the AMP achievable MSE reveals the minimum length of pilot sequences to make the AMP algorithm accomplish the phase transition. We think this is a guide to how long pilot sequences are needed when designing the system. For the prediction error that may occur when using the empirical AMP algorithm in the practice, we can appropriately increase the number of pilots to ensure that the phase transition occurs.

%In addition, the performance of CE under different settings of transmit powers and the number of BS antennas can be seen in Figs. \ref{cept1}-\ref{cem1}. We can derive similar results of the phase transitions in CE compared of the analysis of AUD performances, and we notice that although the MSE performances tend to be consistent with pilot length increasing under different settings of transmit powers, the NMSE performances will be different. And, we observe that increasing the number of BS antennas will not benefit the MSE performance in the region after the phase transition.
\section{Spatially Correlated Channel}\label{spc}
In this section, we concern another typical channel scenario in the wireless communications literature, where the scattering is localized around the user devices and the BS is elevated and thus has no scatterers in its near filed \cite{MassiveMIMON2017}, resulting in the spatially correlated channel. Assuming no line-of-sight propagation, the channel of each user $k$ in the group $g$ is distributed as $Q_g(\mathbf h_{g}) = \mathcal{CN}(\boldsymbol 0, \mathcal C_g)$, where $\mathcal C_g = \boldsymbol U_g\boldsymbol\Lambda_g\boldsymbol U_g^H$ with a rank $r_g\ll M$. Considering the number of BS antennas is sufficiently large, the $\boldsymbol\Lambda_g$ is approximated diagonal \cite{CEMoG2015, CoCE2013}.

We further suppose that different user groups are sufficiently well separated in the angle of arrival (AoA) domain and the angular spread (AS) of each group is sufficiently small. Accordingly, we assume that channels coefficients in all the user groups are in different mutually orthogonal subspaces such that $\boldsymbol U_g^H\boldsymbol U_j = \boldsymbol 0$, for $j\neq g$. We note that although directly achieving the mutually orthogonal subspaces is too restrictive, some user scheduling strategies can be adopted to guarantee the user groups in mutually orthogonal subspaces are served simultaneously \cite{JSDM2013}. In the following of this section, we will see that the above spatially correlated channel assumption reduces the expression of the free entropy function (\ref{general_FE}) and provides some novel propositions for joint AUD and CE.

\subsection{Free Entropy Function}
%In this subsection, we derive the free entropy function under the assumptions in spatial correlated channel scenario based on the general free entropy equation (\ref{general_FE}), which provides the preliminaries for the analysis of joint AUD and CE performances.
%\begin{equation}\mathcal E_g \triangleq \frac{1}{G}\sum_{g=1}^G\left(\hat{\boldsymbol x}_{g}(\boldsymbol Y)-\boldsymbol s_{g}\right)\left(\hat{\boldsymbol x}_{g}(\boldsymbol Y)-\boldsymbol s_{g}\right)^H\end{equation}
Before deriving the free entropy function in the spatially correlated channel scenario, we recall that the matrix $\mathcal E$ is defined as $\mathcal EG= \sum_g \mathcal E_g$, where $\mathcal E_g $ is the corresponding recovery MSE matrix of each user group $g$. As a consequence, we have the following lemma.
\begin{lemma}
Under the spatially correlated channel assumption, i.e., $\boldsymbol U_g^H\boldsymbol U_j = \boldsymbol 0$, for $j\neq g$, all the stationary points of the free entropy function (\ref{general_FE}) fulfill the equation: $\mathcal E^{(\text{s})} =
\sum_g\boldsymbol U_g\boldsymbol\Xi_g^{(\text{s})}\boldsymbol U_g^H$ for some symmetric positive semidefinite matrix $\boldsymbol\Xi_g^{(\text{s})}$.
\end{lemma}
\begin{proof}
Please see the appendix \ref{L_1}.
\end{proof}
Since the statistical characteristics of our considered scenario are all reflected in the stationary points of the free entropy function, we restrict the form of the matrix $\mathcal E$ as $\mathcal E =
\sum_g\boldsymbol U_g\boldsymbol\Xi_g\boldsymbol U_g^H$. As a consequence, the following theorem holds.
\begin{theorem}
Under the spatially correlated channel assumption, the free entropy function (\ref{general_FE}) can be decoupled into a summation of $G$ independent free entropy functions, i.e., $\Phi(\mathcal E) = \sum_{g = 0}^G\Phi_g(\mathcal E_g)$, where $\Phi_0(\mathcal E_0)$ is a constant, and the specific expression of each $\Phi_g(\mathcal E_g)$ is formulated in the equation (\ref{d_free}).

\end{theorem}
\begin{proof}
Please see the appendix \ref{T_3}.
\end{proof}
Accordingly, the recovery MSE of each user group can be evaluated separately via the equation (\ref{d_free}), indicating that the transmitted signals from users outside one group will not affect the reconstruction within such a group. In the following, the performance analysis of joint AUD and CE in the spatially correlated channel scenario will be provided based on the \emph{Theorem 3}.

\begin{figure*}
\begin{center}
\subfigure[$\alpha = 0.11$.]{
\begin{minipage}[t]{0.31\linewidth}
\centering
\includegraphics[width=2.3in]{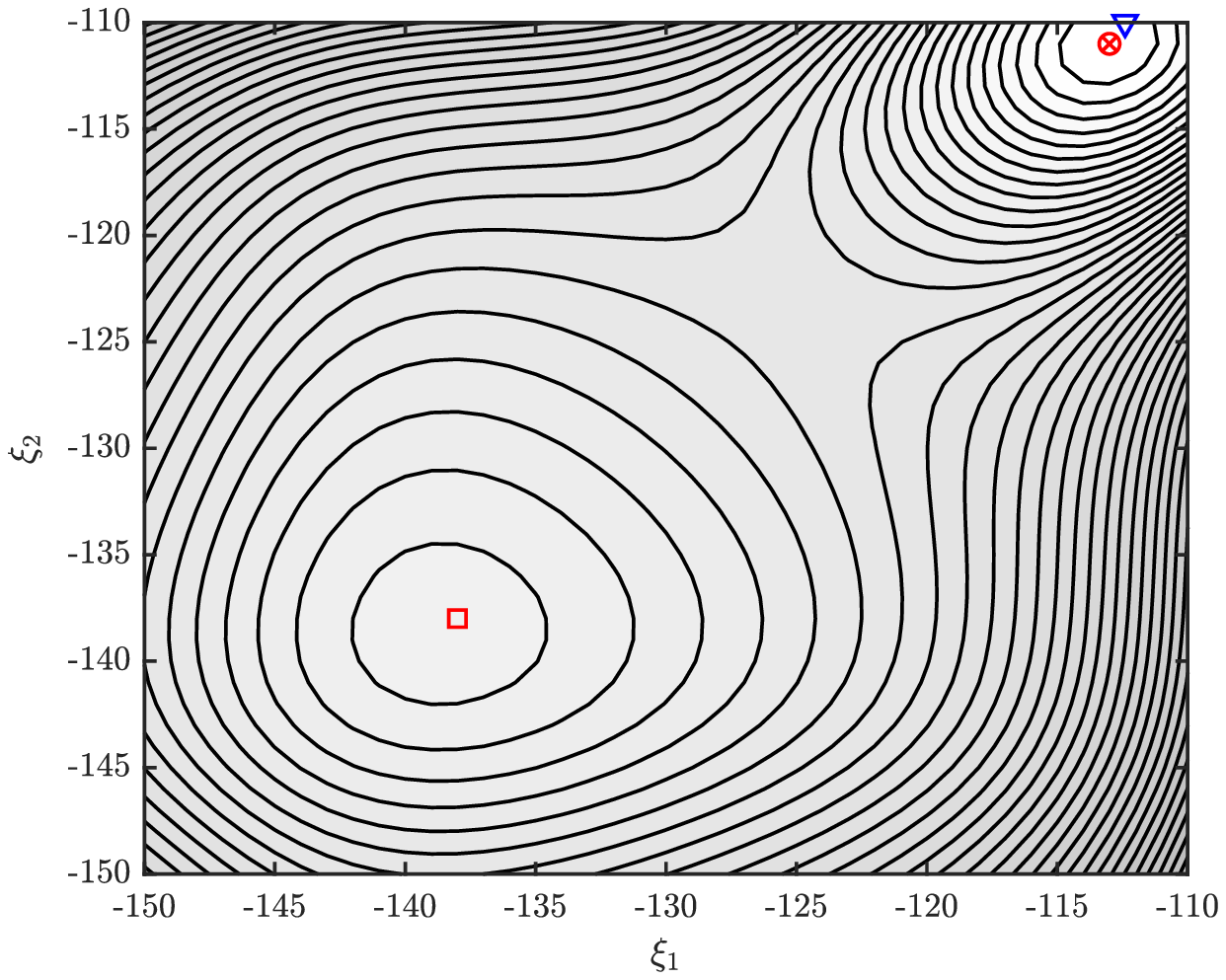}
%\caption{$P_t = 18$dBm.}
\label{al011}
\end{minipage}%
}
\subfigure[$\alpha = 0.13$.]{
\begin{minipage}[t]{0.31\linewidth}
\centering
\includegraphics[width=2.3in]{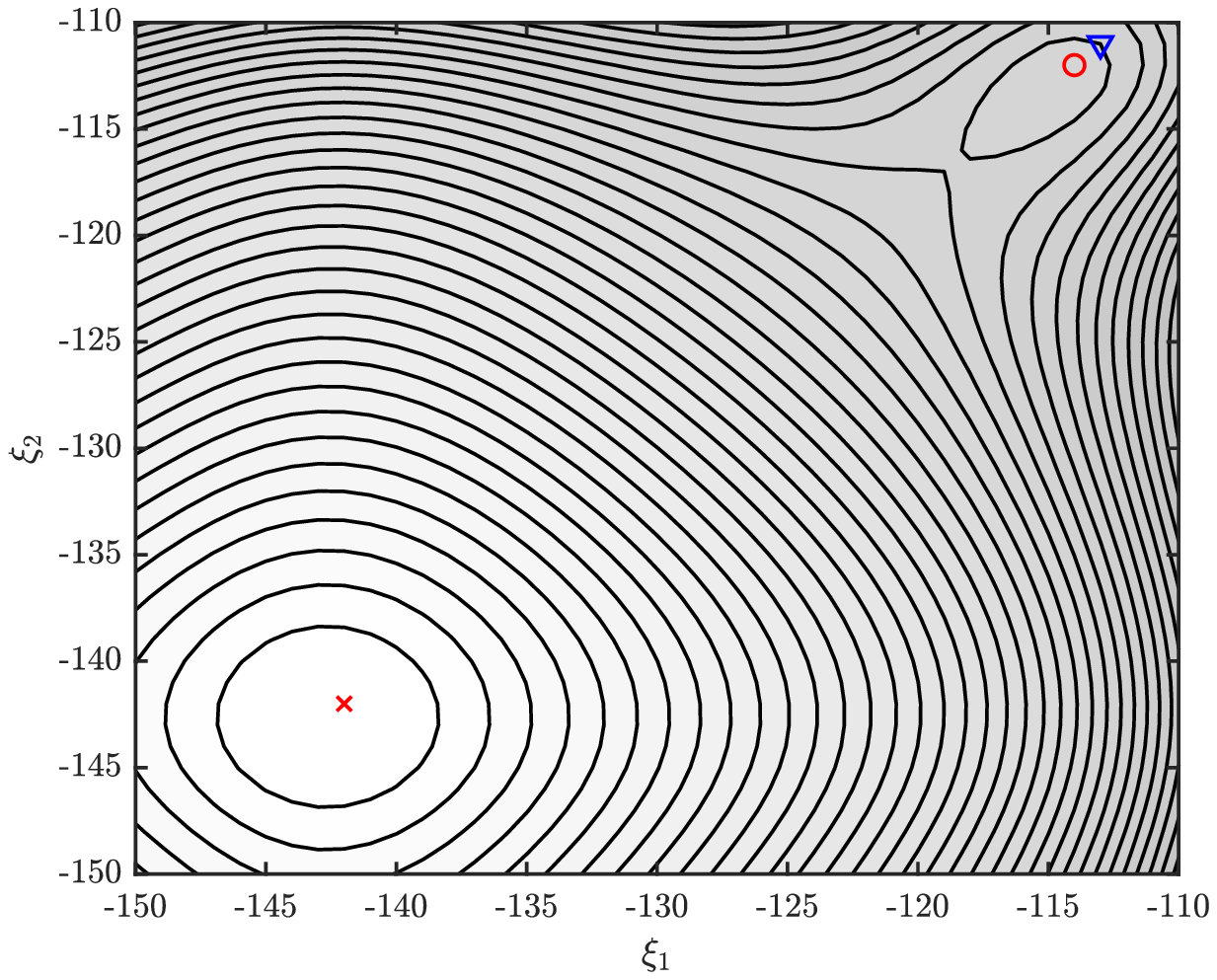}
%\caption{fig2}
\label{al013}
\end{minipage}
}
\subfigure[$\alpha = 0.15$.]{
\begin{minipage}[t]{0.31\linewidth}
\centering
\includegraphics[width=2.3in]{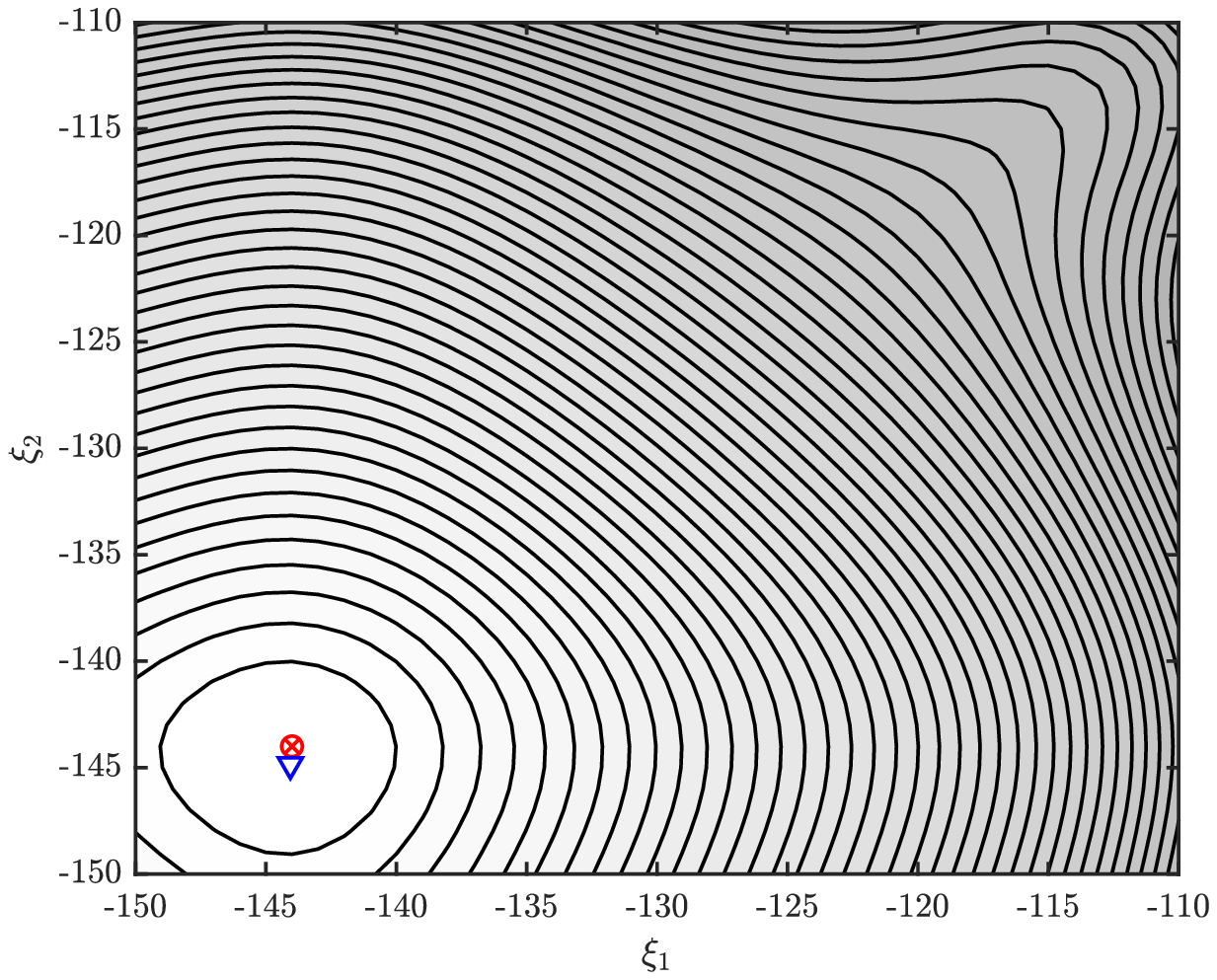}
%\caption{fig3}
\label{al015}
\end{minipage}
}
\caption{Free entropy as a function of MSE with $P_t = 18$dBm transmit power under different settings of $\alpha$ in spatially correlated channel scenario. (The ``red cross'' denotes the MMSE point, the ``red circle'' denotes the predicted AMP-achievable MSE point, the ``red square'' denotes the unachievable local maxima MSE point, and the ``blue triangle'' denotes the simulated AMP MSE point.)}
\end{center}
\end{figure*}

%\setcounter{TempEqCnt}{\value{equation}}
%\setcounter{equation}{23}
%\begin{figure*}[tbp]
%\begin{align}
%\tilde{\Phi}_g(\boldsymbol\Xi_g) = &-\alpha\sum\limits_{m = 1}^{r_g}\left(\frac{\sigma_w^2}{\frac{1}{\alpha}\xi_{g,m}+\sigma_w^2}+\log\left(\frac{1}{\alpha}\xi_{g,m}+\sigma_w^2\right)\right)+\sum\limits_{m = 1}^{r_g}\frac{(1-\rho)\lambda_{g,m}}{\lambda_{g,m}+\sigma_w^2+\frac{1}{\alpha}\xi_{g,m}}\nonumber\\
%&+\int{\rm D}\boldsymbol z\rho\log\left[(1-\rho)\prod_{m = 1}^{r_g}\exp\left\{-\frac{|z_m|^2\lambda_{g,m}}{\sigma_w^2+\frac{1}{\alpha}\xi_{g,m}}\right\}+\rho\prod_{m = 1}^{r_g}\frac{\frac{1}{\alpha}\xi_{g,m}+\sigma_w^2}{\frac{1}{\alpha}\xi_{g,m}+\sigma_w^2+\lambda_{g,m}}\right]\nonumber\\
%&+\int{\rm D}\boldsymbol z(1-\rho)\log\left[(1-\rho)\prod_{m = 1}^{r_g}\exp\left\{-\frac{|z_m|^2\lambda_{g,m}}{\lambda_{g,m}+\sigma_w^2+\frac{1}{\alpha}\xi_{g,m}}\right\}+\rho\prod_{m = 1}^{r_g}\left(\frac{\frac{1}{\alpha}\xi_{g,m}+\sigma_w^2}{\frac{1}{\alpha}\xi_{g,m}+\sigma_w^2+\lambda_{g,m}^2}\right)\right].\label{fe_si}
%\end{align}
%\hrulefill
%\end{figure*}
%\setcounter{equation}{\value{TempEqCnt}}

\subsection{Prediction of AMP-Based joint AUD and CE}
According to the \emph{Lemma 1}, we define $\boldsymbol\Xi_g^{\star}, \forall g$ as the matrixs that associated with the fixed point $\mathcal E^{\star}$ of AMP state evolution. We further restrict the structure of the matrix $\boldsymbol \Xi_g^{\star}$ to be diagonal\footnote{Restricting the matrix $\boldsymbol\Xi_g^{\star}$ into diagonal is reasonable. As proved in \cite{RepMMVG2018}, the state of state evolution function always maintains the same structure with iterations.  The initial state of the state evolution in the AMP framework is always set to be $\mathcal E_g^{(0)} = \mathbb E(\boldsymbol s_g\boldsymbol s_g^H)$, and in our situation, we have $\mathcal E_g^{(0)} = \rho \mathcal C_g = \rho\boldsymbol U_g\boldsymbol \Lambda_g\boldsymbol U_g^H$. Hence, $\boldsymbol\Xi_g^{\star}$ will also be diagonal since the matrix $\boldsymbol\Lambda_g$ is approximately diagonal in our scenario with a large number of antennas.}, with the $m$th diagonal element denoted as $\xi_{g,m}^{\star}$. By combining \emph{Theorem 3} with \emph{Lemma 1}, the matrix $\boldsymbol\Xi_g^{\star}$ can be obtained by seeking the local maximum point of the equation (\ref{fe_si}), where the quantities $\xi_{g,m}$, $\lambda_{g,m}$ are defined as the $m$th element of the diagonal of $\boldsymbol\Xi_g$ and $\boldsymbol\Lambda_g$.

Then, we consider the performance analysis of joint AUD and CE. For AUD, defining $\underline{\boldsymbol s_g}\triangleq \boldsymbol U_g^H\hat{\boldsymbol s}_g$ and together with (\ref{llr}), the detection sufficient statistic in such a scenario can be formulated as
\setcounter{equation}{22}
\begin{equation}
\text{T}(\underline{\boldsymbol s_g}) = \sum\limits_{m = 1}^{r_g}\frac{\lambda_{g,m}|\underline{s_{g,m}}|^2}{\left(\sigma_w^2+\frac{1}{\alpha}\xi_{g,m}^{\star}\right)\left(\sigma_w^2+\frac{1}{\alpha}\xi_{g,m}^{\star}+\lambda_{g,m}\right)}.\label{det2}
\end{equation}
and the involved conditional probabilities of $\underline{\boldsymbol s_g}$ are given as
\begin{align}
p(\underline{\boldsymbol s_g}|a_g = 1) &= \mathcal{CN}(\underline{\boldsymbol s_g}; \boldsymbol 0, \boldsymbol\Lambda_g+\boldsymbol\Delta+\frac{1}{\alpha}\boldsymbol\Xi_g^{\star}),\nonumber\\
p(\underline{\boldsymbol s_g}|a_g = 0) &= \mathcal{CN}(\underline{\boldsymbol s_g}; \boldsymbol 0, \boldsymbol\Delta+\frac{1}{\alpha}\boldsymbol\Xi_g^{\star}).\label{CP}
\end{align}

As a result, we can then derive the specific functional forms of the detection performance metrics in the following proposition.
\begin{proposition}
The missed detection probability with the definition $P_D = {\rm Pr}\{\text{T}(\underline{\boldsymbol s_g})>l_g|a_g = 1\}$ and the false alarm probability with the definition $P_F = {\rm Pr}\{\text{T}(\underline{\boldsymbol s_g})>l_g|a_g = 0\}$ based on the conditional probabilities (\ref{CP}) and the detection sufficient statistic (\ref{det2}) are given as
\begin{align}
P_D = \sum\limits_{m = 1}^{r_g}\prod_{j\neq m}\frac{\omega_{g,j}}{\omega_{g,j}-\omega_{g,m}}\exp\{-\omega_{g,m}l_g\},\nonumber\\
P_F = \sum\limits_{m = 1}^{r_g}\prod_{j\neq m}\frac{\tilde{\omega}_{g,j}}{{\tilde\omega}_{g,j}-{\tilde\omega}_{g,m}}\exp\{-{\tilde\omega}_{g,m}l_g\},\label{pro2}
\end{align}
where $$\omega_{g,m}\triangleq (\sigma_w^2+\frac{1}{\alpha}\xi_{g,m}^{\star})\lambda_{g,m}^{-1},$$ and $$\tilde{\omega}_{g,m}\triangleq (\sigma_w^2+\frac{1}{\alpha}\xi_{g,m}^{\star}+\lambda_{g,m})\lambda_{g,m}^{-1}.$$
\end{proposition}

\begin{proof}
Please see the appendix \ref{dp}.
\end{proof}
%For AUD, from (\ref{pf}) and (\ref{pm}), the false alarm and missed detection probabilities are determined by the detection sufficient statistic. Combining the \emph{Lemma 1} with the general expression of LLR in (\ref{llr}) and defining $\boldsymbol\Xi_g^{\star}$ as the matrix that associated with the fixed point $\mathcal E^{\star}$ of AMP state evolution, the detection sufficient statistic in the spatially correlated channel scenario can be formulated as
%\begin{align}
%\text{T}(\hat{\boldsymbol s}_g) = & \hat{\boldsymbol s}_g^H\left(\boldsymbol\Sigma^{-1}-\left(\mathcal C_g+\boldsymbol\Sigma\right)^{-1}\right)\hat{\boldsymbol s}_g\nonumber\\ = &\hat{\boldsymbol s}_g^H\left(\left(\boldsymbol\Delta+\frac{1}{\alpha}\sum\limits_g\boldsymbol U_g\boldsymbol \Xi_g^{\star}\boldsymbol U_g^H\right)^{-1}\right.\nonumber\\
%&\quad\left.-\left(\mathcal C_g+\boldsymbol\Delta+\frac{1}{\alpha}\sum\limits_g\boldsymbol U_g\boldsymbol \Xi_g^{\star}\boldsymbol U_g^H\right)^{-1}\right)\hat{\boldsymbol s}_g\nonumber\\
% = &\hat{\boldsymbol s}_g^H\left(\boldsymbol U_g\left(\sigma_w^2\mathbf I+\frac{1}{\alpha}\boldsymbol \Xi_g^{\star}\right)^{-1}\boldsymbol U_g^H\right.\nonumber\\
%&\quad\left.-\boldsymbol U_g\left(\sigma_w^2\mathbf I+\boldsymbol \Lambda_g +\frac{1}{\alpha}\boldsymbol \Xi_g^{\star}\right)^{-1}\boldsymbol U_g^H\right)\hat{\boldsymbol s}_g.\label{Dete_1}
%\end{align}
As for the analysis for CE, according to (\ref{cha_err_ma}), the CE error of the user group $g$ associated with the $m$th  eigenvalue $\lambda_{g,m}$ can be formulated as $(\sigma_w^2+\lambda_{g,m} +\frac{1}{\alpha}\xi_{g,m}^{\star})^{-1}$.

After that, we can see that the performance of joint AUD and CE can be predicted by the local maximum point of the free entropy function (\ref{fe_si}). The verification of such a result, a phase transition diagram and an optimality analysis of joint AUD and CE in the spatially correlated channel scenario will be shown in the following Sec. \ref{ve_sc}.

\subsection{Discussion}
Before proceeding, we provide some discussions about our theoretical results. Multiplying $\boldsymbol U_g^{*}$ for each group by the right of the received signal $\boldsymbol Y$ in the original model (\ref{channel}), we can obtain $G$ a sub-model expressed as
\begin{equation}
    \underline{\boldsymbol Y}_g = \underline{\boldsymbol F}_g\underline{\boldsymbol X}_g+\underline{\boldsymbol W}_g,\label{red_mo}
    \end{equation}
    where $\underline{\boldsymbol Y}_g\triangleq \boldsymbol Y\boldsymbol U_g^{*}$, $\underline{\boldsymbol X}_g\triangleq \boldsymbol X\boldsymbol U_g^{*}$, $\underline{\boldsymbol W}_g\triangleq \boldsymbol W\boldsymbol U_g^{*}$ and $\underline{\boldsymbol F}_g$ is the corresponding pilot matrix for the user device for the group $g$. We note that the free entropy function of (\ref{red_mo}) reflecting its statistical performances corresponds to (\ref{fe_si}), so that separately using model (\ref{red_mo}) for each group is equivalent to jointly processing all the groups with (\ref{channel}), if the mutually orthogonal subspaces condition holds. Such a per-group processing (PGP) idea has been successfully adopted in the case of downlink transmission \cite{JSDM2013} and our theoretical results provide the theoretical foundations for the PGP strategy in the uplink joint AUD and CE problem.

    Complementally, when the mutually orthogonal condition does not perfectly hold, by carefully designing the processing matrix for each group, the PGP strategy can also be applied by permitting the affordable inter-group inference. According to the approach in \cite{JSDM2013}, when the BS equipped with a uniform linear array, the inter-group inference can be absorbed into the additive noise matrix with independent elements. Hence, by slightly modifying the free entropy function (\ref{fe_si}), our framework can also be adopted to evaluate performance of joint AUD and CE, where the mutually orthogonal condition does not perfectly hold.
\subsection{Verification in the Spatially Correlated Channel Scenario}\label{ve_sc}
In this section, we provide numerical results to verify the analysis based on the replica method in spatially correlated channel scenario. We consider the channel coefficients of all user groups are in orthogonal subspaces. Without loss of generality, we investigate the performance analysis results for one certain user group. We consider the typical spatially correlated channel model in \cite{MassiveMIMON2017}. For visualization, we assume the number of BS antennas is $M = 64$. The path loss model of the wireless channel, the distribution of the distance between the BS station and the user devices, the bandwidth the power spectral density and the total number of user devices are considered same with that in the settings of the isotropic channel case.

\subsubsection{Phase Transition of the Free Entropy Function}

We first consider the center angle of user angle domain is set as $45^{\circ}$ distributed as a normal distribution with a $1^{\circ}$ AS, where most of the power of channel coefficients is concentrated on a two-dimensional subspace.
Figs. \ref{al011}-\ref{al015} shows the different MSE performance regions of the free entropy function (\ref{fe_si}) in different settings. We can notice that the performance regions in the spatially correlated channel case is similar with that in the isotropic channel case. Specifically, Fig. \ref{al013} demonstrates the performance gap between the AMP-achievable MSE and the MMSE. The phase transitions appear in both our prediction and the simulated AMP algorithm. Differently from the isotropic channel scenario, it is observed that although there exists a large number of antennas in the BS, the phase transition phenomenon still exists. Fig. 6 also shows that our analytical MSE results predicted in the large system limit and AMP recovery MSE in practical settings, i.e., $T = 220, 260, 300$ and $K = 2000$ are consistent. Such results are evidences that our analytical results are valid when the length of pilot sequences and the number of user devices are large but finite.

%This is because due to the spatial correlation in channel coefficient, i.e., most power of the channel coefficients concentrates on a low dimensional subspace.
\begin{figure*}
\begin{center}
\subfigure[$P_t = 13$dBm.]{
\begin{minipage}[t]{0.31\linewidth}
\centering
\includegraphics[width=2.2in]{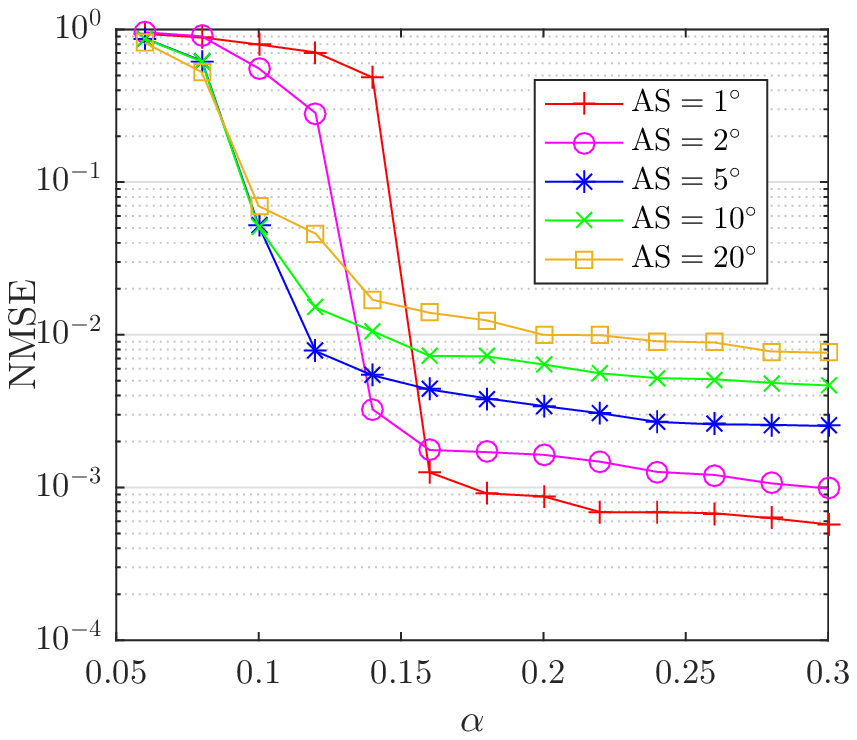}
%\caption{$P_t = 18$dBm.}
%\label{pra13}
\end{minipage}%
}
\subfigure[$P_t = 23$dBm.]{
\begin{minipage}[t]{0.31\linewidth}
\centering
\includegraphics[width=2.2in]{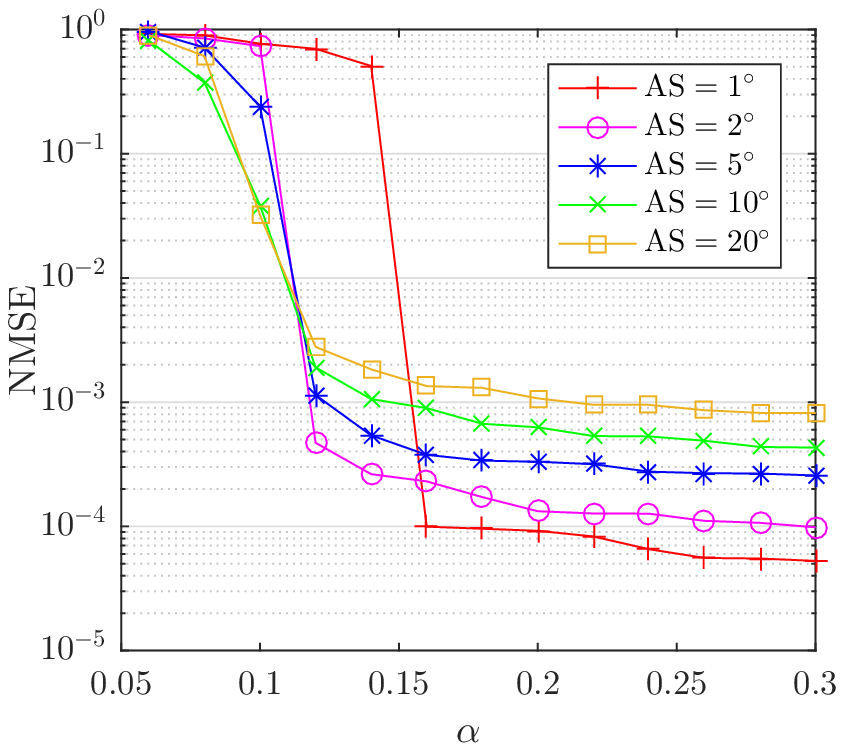}
%\caption{fig2}
%\label{pra23}
\end{minipage}
}
\subfigure[$P_t = 33$dBm.]{
\begin{minipage}[t]{0.31\linewidth}
\centering
\includegraphics[width=2.2in]{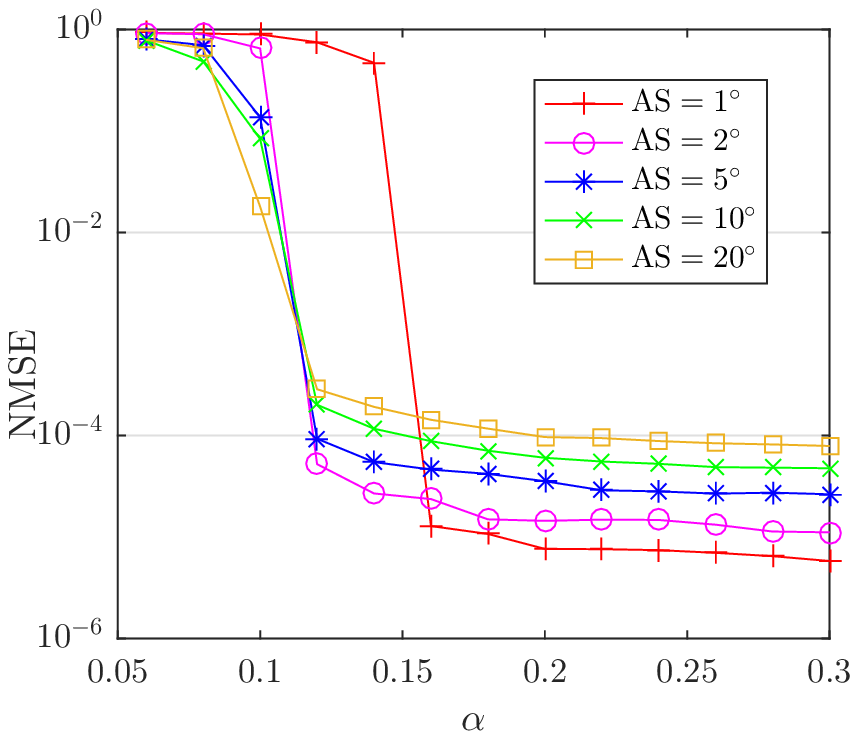}
%\caption{fig3}
%\label{pra33}
\end{minipage}
}
\caption{NMSE performance of AMP versus pilot length in different settings of transmit power and AS.}
\label{pra}
\end{center}
\end{figure*}

Then, we demonstrate the case where the AS is not very small, which is more in line with the practical scenario. The phase transition phenomenons in all the considering settings can be observed in Fig. \ref{pra}. With the transmit power increasing, the phase transition phenomenon is becoming more and more obvious. In addition, with the AS increases, the number of the efficient dimensions in the subspace of the user channel becomes larger and the total energy will be allocated to these effective dimensions, leading that the phase transition phenomenon is weakened. Such an observation matches our previous results that the spatially correlated channel is more likely to promote a phase transition compared with the isotropic channel. Another observation is that smaller AS is beneficial to NMSE performance when pilot resources are abundant, while the pilot resources are scarce, the opposite is true.
\subsubsection{Performance Analysis of AUD}
\begin{figure}
%\begin{minipage}[t]{0.5\linewidth}
\centering
\includegraphics[width=3.5in]{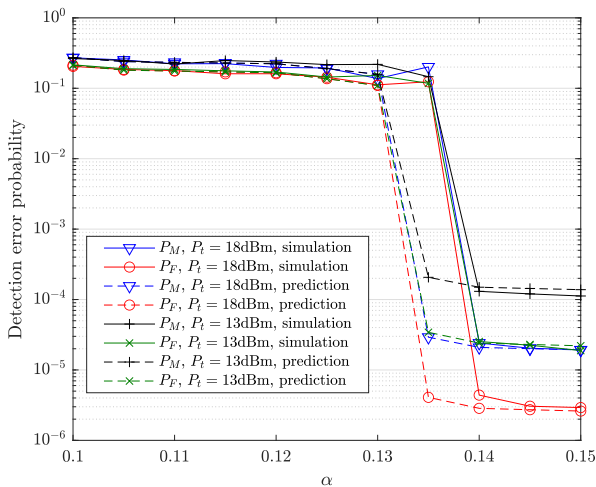}
\caption{Detection performance prediction.}
\label{detpre2}
\end{figure}
%\end{minipage}\quad
%\begin{minipage}[t]{0.5\linewidth}
%\begin{figure}
%\centering
%\includegraphics[width=3.5in]{fig/det_perfpt2.eps}
%\caption{Detection performance as a function of transmit power.}
%\label{detpt2}
%%\end{minipage}
%\end{figure}
%

Fig. \ref{detpre2} depicts the prediction of AMP-based AUD as a function of pilot length with $1^{\circ}$ AS. It is observed that the phase transition of the detection error probabilities appears in both $P_t = 13$dBm and $P_t = 18$dBm settings, which is different from the results in the isotropic channel case. This is because the energy of a large number of antennas is concentrated in a small amount of subspace dimensions, increasing the signal-to-noise ratio in considered dimensions.

%Fig. \ref{detpt2} depicts the AUD performance as a function of the transmit power in three settings, i.e., $\alpha = 0.13$, $0.14$ and $0.15$. It is observed that in $\alpha = 0.13$ case, the detection error probabilities will not decrease as the transmit power increases. Differently, in the case of $\alpha = 0.14$, we see that in the region of $P_t\ge2$dBm, the detection error probabilities decrease fast with the power increasing. Intuitively, the phase transition occurs between $\alpha = 0.13$ and $\alpha = 0.14$ in the region of a large power. Differently from the isotropic channel case, with spatially correlated channel, the phase transition discards only in the region with very small transmit power. In addition, we note that further increasing the pilot length can not obviously improve the detection performances.

In summary, we note that the spatially correlated channel promotes the appearance of phase transition compared with the isotropic channel. This is because the energy of a large number of antennas is concentrated in a small amount of subspace dimensions. This reduces the effective dimensions of the BS antennas, while increasing the power on each dimension. Despite that, the pilot length required to step over the phase transition is highly reduced in the spatially correlated channel scenario, i.e., about $\alpha = 0.14$ for the spatially correlated channel case and about $\alpha = 0.625$ for the isotropic channel case with $5$ user groups.

\section{Conclusion}
In this paper, we have provided an analysis of joint AUD and CE problem under massive connectivity based on the replica method. Particularly, we have established a novel theoretical framework under a general channel model that reduces to multiple typical MIMO channel models. Based on the general framework, we have analyzed two typical scenarios in massive connectivity, i.e., the isotropic channel case and the spatially correlated channel case. We have provided the analysis of the Bayes-optimility, phase transition, and the predictions of the performance of joint AUD and CE in both the two cases. In addition, we have shown that the spatially correlated channel is more likely to promote a phase transition compared with the isotropic channel, due to the small number of subspace dimensions and corresponding antenna gains. However, thanks to the spatially correlated channel, where all the user group can be perfectly partitioned, the pilot length required to step over the phase transition can be significantly reduced.

Some future directions of research are also implied by this paper. 1) Our theoretical framework was established in the Bayes-optimal condition, if the condition is not met, what changes will exist in the analytical results. 2) Exploiting the potential user grouping information may bring some advantages, and how to design an algorithm to achieve joint user grouping, AUD and CE. 3) How to design the processing matrix in the PGP-based strategy for joint AUD and CE.

%In the isotropic case, we have analyzed the different performance regions of AMP-based joint AUD and CE, demonstrating that the region where AMP is sub-optimal persists but becomes small and eventually varnishes as the number of antennas increases and the transmit power decreases. We have shown that the phase transition exists for small numbers of antennas and disappears in the large-scale antenna case. As for the spatially correlated channel case, we have proven that with an orthogonal subspace of user group channels, the transmitted signals from users outside one group will not affect the joint AUD and CE performance of the users within such a group.

\appendix
\subsection{Proof of Theorem 1}
\label{A}

According to the expression of partition function $Z$ in  (\ref{patition}), we notice that $Z^{n}$ can be written as
\begin{align}
Z^{n} = \int &\prod_{a} {\rm d} \boldsymbol X^{(a)}\prod_{a,g,k}\left((1-\rho)\delta(\boldsymbol x_{gk}^{(a)})+\rho Q_g(\boldsymbol x_{gk}^{(a)})\right)\nonumber\\
&\times\prod_{t}\frac{1}{\pi^{nM}|\boldsymbol\Delta|^n}e^{-\sum\limits_{a}(\boldsymbol y_t-\boldsymbol f_t\boldsymbol X^{(a)})\boldsymbol\Delta^{-1}(\boldsymbol y_t-\boldsymbol f_t\boldsymbol X^{(a)})^{H}}\nonumber,
\end{align}

\setcounter{TempEqCnt}{\value{equation}}
\setcounter{equation}{29}
\begin{figure*}[tbp]
\begin{small}
\begin{align}
1& =  \int{\rm d}\boldsymbol\Theta\exp\left(\sum\limits_{g}{\rm Tr}\left(\sum\limits_{a}\hat{\mathcal Q}_g^{(a)H}\left(\frac{K}{2}\mathcal Q_g^{(a)}-\frac{1}{2}\sum\limits_{k}\mathbf x_{kg}^{(a)}\mathbf x_{kg}^{(a)^H}\right)+\sum\limits_{a}\hat{\mathcal Q}_g^{(a)T}\left(\frac{K}{2}\mathcal Q_g^{(a)*}-\frac{1}{2}\sum\limits_{k}\mathbf x_{kg}^{(a)*}\mathbf x_{kg}^{(a)^T}\right)\right.\right.\nonumber\\
&\left.-\sum\limits_{a\neq b}\hat{\mathcal T}_g^{(ab)H}\left(\frac{K}{2}\mathcal T_g^{(ab)}-\frac{1}{2}\sum\limits_{k}\mathbf x_{kg}^{(a)}\mathbf x_{kg}^{(b)^H}\right)\right.-\left.\sum\limits_{a\neq b}\hat{\mathcal T}_g^{(ab)T}\left(\frac{K}{2}\mathcal T_g^{(ab)*}-\frac{1}{2}\sum\limits_{k}\mathbf x_{kg}^{(a)*}\mathbf x_{kg}^{(b)^T}\right)\right.\nonumber\\
&\left.\left.-\sum\limits_{a}\hat{\mathcal M}_g^{(a)H}\left(K\mathcal M_g^{(a)}-\sum\limits_{k}\mathbf x_{kg}^{(a)}\mathbf s_{kg}^H\right)-\sum\limits_{a}\hat{\mathcal M}_g^{(a)T}\left(K\mathcal M_g^{(a)*}-\sum\limits_{k}\mathbf x_{kg}^{(a)*}\mathbf s_{kg}^T\right)\right)\right),\label{id}
\end{align}
\setcounter{equation}{31}
\begin{align}
&\mathbb EZ^{n} = \int{\rm d}\boldsymbol\Theta\exp\bigg(Kn\sum\limits_{g}{\rm Tr}\left(\frac{1}{2}\left(\hat{\mathcal Q}_g^{H}\mathcal Q_g+\hat{\mathcal Q}_g^{T}\mathcal Q_g^{*}\right)-\frac{1}{2}(n-1)\left(\hat{\mathcal T}_g^{H}\mathcal T_g+\hat{\mathcal T}_g^{T}\mathcal T_g^{*}\right)\right.-\left.\left(\hat{\mathcal M}_g^{H}\mathcal M_g+\hat{\mathcal M}_g^{T}\mathcal M_g^{*}\right)\right)\bigg)\nonumber\\
&\times\prod_{t}\frac{\beta_t}
{\pi^{Mn}|\boldsymbol\Delta|^n}\prod_{g}\bigg(\int{\rm d}\boldsymbol s_{g}((1-\rho)\delta(\boldsymbol s_{g})+\rho Q_g(\boldsymbol s_{g}))\int{\rm D}\boldsymbol z\bigg(\int\left({\rm d}\boldsymbol x_{g}((1-\rho)\delta(\boldsymbol x_{g})+\rho Q_g(\boldsymbol x_{g}))\right)\nonumber\\
&\times\exp\left(-\frac{1}{2}\boldsymbol x_g^{H}(\hat{\mathcal Q}_g^{H}+\hat{\mathcal Q}_g+\hat{\mathcal T}_g^{H}+\hat{\mathcal T}_g)\boldsymbol x_g\right.+\boldsymbol x_g^{H}\big(\big(\frac{\hat{\mathcal T}_g^{H}+\hat{\mathcal T}_g}{2}\big)^{\frac{1}{2}}\boldsymbol z+\hat{\mathcal M}_g\boldsymbol s_g\big)+\big(\boldsymbol z^H\big(\frac{\hat{\mathcal T}_g^{H}+\hat{\mathcal T}_g}{2}\big)^{\frac{1}{2}}+\boldsymbol s_g^{H}\hat{\mathcal M}_g^{H}\big)\boldsymbol x_g\bigg)\bigg)^n\bigg)^K.\label{ezhs}
\end{align}
\begin{align}
\beta_t = &\left|\mathbf I_{Mn}+\sigma_w^{-2}\mathcal G\right|^{-1}\nonumber\\
 = &\left|n\left(\boldsymbol\Delta+\frac{K}{T}\sum\limits_g\left(\mathcal Q_g-\mathcal T_g\right)\right)^{-1}\left(\frac{K}{T}\left(\sum\limits_g\rho\mathcal C_g-\mathcal M_g-\mathcal M_g^{H}+\mathcal T_g\right)+\boldsymbol\Delta\right)+\mathbf I_{M}\right|^{-1}\times\left|\mathbf I_M+\sigma_w^{-2}\frac{K}{T}\sum\limits_{g}(\mathcal Q_g-\mathcal T_g)\right|^{-n}\nonumber\\
\approx &\exp \bigg(-n{\rm Tr}\bigg(\bigg(\boldsymbol\Delta+\frac{K}{T}\sum\limits_g\bigg(\mathcal Q_g-\mathcal T_g\bigg)\bigg)^{-1}\bigg(\frac{K}{T}\bigg(\sum\limits_g\rho\mathcal C_g-\mathcal M_g-\mathcal M_g^{H}+\mathcal T_g\bigg)+\boldsymbol\Delta\bigg)\bigg)\nonumber\\
&\quad\quad\quad-n\log\left|\mathbf I_M+\sigma_w^{-2}\frac{K}{T}\sum\limits_{g}(\mathcal Q_g-\mathcal T_g)\right|\bigg)\label{bt}.
\end{align}
\setcounter{equation}{34}
\begin{align}
&\tilde{\Phi}(\boldsymbol{\mathcal Q}, \boldsymbol{\mathcal T}, \boldsymbol{\mathcal M}, \boldsymbol{\hat{\mathcal Q}}, \boldsymbol{\hat{\mathcal T}}, \boldsymbol{\hat{\mathcal M}}) =\sum\limits_{g}{\rm Tr}\left(\frac{1}{2}\left(\hat{\mathcal Q}_g^{H}\mathcal Q_g+\hat{\mathcal Q}_g^{T}\mathcal Q_g^{*}\right)+\frac{1}{2}\left(\hat{\mathcal T}_g^{H}\mathcal T_g+\hat{\mathcal T}_g^{T}\mathcal T_g^{*}\right)-\left(\hat{\mathcal M}_g^{H}\mathcal M_g+\hat{\mathcal M}_g^{T}\mathcal M_g^{*}\right)\right)\nonumber\\
&-\alpha{\rm Tr}\bigg(\bigg(\boldsymbol\Delta+\frac{1}{\alpha}\sum\limits_g\bigg(\mathcal Q_g-\mathcal T_g\bigg)\bigg)^{-1}\bigg(\frac{1}{\alpha}\bigg(\sum\limits_g\rho\mathcal C_g-\mathcal M_g-\mathcal M_g^{H}+\mathcal T_g\bigg)+\boldsymbol\Delta\bigg)\bigg)-\alpha\log\left|\boldsymbol\Delta+\frac{1}{\alpha}\sum\limits_{g}(\mathcal Q_g-\mathcal T_g)\right|\nonumber\\
&+\sum\limits_{g}\bigg(\int{\rm d}\boldsymbol s_{g}((1-\rho)\delta(\boldsymbol s_{g})+\rho Q_g(\boldsymbol s_{g}))\int{\rm D}\boldsymbol z\log\bigg(\int\left({\rm d}\boldsymbol x_{g}((1-\rho)\delta(\boldsymbol x_{g})+\rho Q_g(\boldsymbol x_{g}))\right)\exp\left(-\frac{1}{2}\boldsymbol x_g^{H}(\hat{\mathcal Q}_g^{H}+\hat{\mathcal Q}_g\right.\nonumber\\
&+\hat{\mathcal T}_g^{H}+\hat{\mathcal T}_g)\boldsymbol x_g
+\boldsymbol x_g^{H}\big(\big(\frac{\hat{\mathcal T}_g^{H}+\hat{\mathcal T}_g}{2}\big)^{\frac{1}{2}}\boldsymbol z+\hat{\mathcal M}_g\boldsymbol s_g\big)+\big(\boldsymbol z^H\big(\frac{\hat{\mathcal T}_g^{H}+\hat{\mathcal T}_g}{2}\big)^{\frac{1}{2}}+\boldsymbol s_g^{H}\hat{\mathcal M}_g^{H}\big)\boldsymbol x_g\bigg)\bigg)\bigg)-\alpha\log\pi^M.\label{phi0}
\end{align}
\end{small}
\hrulefill
\end{figure*}
\setcounter{equation}{\value{TempEqCnt}}

As a consequence, by considering the equation (\ref{channel}), the average free entropy can be derived as
\begin{align}
&\mathbb EZ^{n} = \int \prod_{a} {\rm d} \boldsymbol X^{(a)}\prod_{a,g,k}\left((1-\rho)\delta(\boldsymbol x_{gk}^{(a)})+\rho Q_g(\boldsymbol x_{gk}^{(a)})\right)\nonumber\\
&\times\int\prod_{g,k}{\rm d}\boldsymbol s_{gk}\left((1-\rho)\delta(\boldsymbol s_{gk})+\rho Q_g(\boldsymbol s_{gk})\right)\prod_{t}\frac{1}{\pi^{nM}|\boldsymbol\Delta|^n}\beta_t.\label{Ezn}%
\end{align}
The expectation is taken over $\boldsymbol F$, $\boldsymbol S$, and $\boldsymbol W$, and the quantity $\beta_t$ is defined as
\begin{align}
\beta_t \triangleq \mathbb E_{\boldsymbol F, \boldsymbol W}\big(\exp(-\sigma_w^{-2}\sum\limits_{a}||\boldsymbol v_t^{(a)}||^2)\big),\nonumber
\end{align}
where we define $$\boldsymbol v_t^{(a)}\triangleq \sum_{g,k}f_{tgk}\boldsymbol s_{gk}-\sum_{g,k}f_{tgk}\boldsymbol x_{gk}^{(a)}+\boldsymbol w_t, a = 1,\dots,n.$$
We note that $f_{tgk}$ is the corresponding element in the vector $\boldsymbol f_t$, and the vector $\boldsymbol w_t$ is the corresponding noise vector. Since we consider a Gaussian random pilot matrix with zero mean, $\boldsymbol v_t^{(a)}$ is still a Gaussian vector with zero mean.
The covariance matrices are calculated as follow.
\begin{small}
\[
\begin{split}
\mathbb E_{\boldsymbol F, \boldsymbol W}(\boldsymbol v_t^{(a)}\boldsymbol v_t^{(a)H}) =& \frac{1}{T}\sum_{g,k}(\boldsymbol s_{kg}-\boldsymbol x_{kg}^{(a)})(\boldsymbol s_{kg}-\boldsymbol x_{kg}^{(a)})^{H}+\boldsymbol\Delta,\\
\mathbb E_{\boldsymbol F, \boldsymbol W}(\boldsymbol v_t^{(a)}\boldsymbol v_t^{(b)H}) =& \frac{1}{T}\sum_{g,k}(\boldsymbol s_{kg}-\boldsymbol x_{kg}^{(a)})(\boldsymbol s_{kg}-\boldsymbol x_{kg}^{(b)})^{H}+\boldsymbol\Delta.
\end{split}
\]
\end{small}
We further define a vector $\tilde{\boldsymbol v}_t\triangleq [\boldsymbol v_t^{(1)},\dots, \boldsymbol v_t^{(n)}]^T$, which is a stitched vector of $\boldsymbol v_t^{(a)}$. It is obvious that the vector $\tilde{\boldsymbol v}_t$ is a Gaussian vector with zero mean, and we define $\mathcal G$ as its covariance matrix. Then, the term $\beta_t$ can be rewritten as
\begin{align}
\mathbb E_{\boldsymbol F, \boldsymbol W}\big(\exp(-\sigma_w^{-2}\sum\limits_{a}||\boldsymbol v_t^{(a)}||^2)\big)= |\mathbf I_{Mn}+\sigma_w^{-2}\mathcal G|^{-1}.\nonumber
\end{align}
We notice that the matrix $\mathcal G$ can be separated into $n\times n$ blocks of size $M\times M$, and the block in the $a$th row and $b$th column satisfies $\mathcal G_{ab} = \mathbb E_{\boldsymbol F, \boldsymbol W}(\boldsymbol v_t^{(a)}\boldsymbol v_t^{(b)H})$. Giving the overlaps
\begin{align}
\rho\mathcal C_g &= \frac{1}{K}\sum\limits_{k = 1}^K\mathbf s_{kg}\mathbf s_{kg}^H,\
\mathcal M_g^{(a)} = \frac{1}{K}\sum\limits_{k = 1}^K\mathbf x_{kg}^{(a)}\mathbf s_{kg}^H,\nonumber\\
\mathcal Q_g^{(a)} &= \frac{1}{K}\sum\limits_{k = 1}^K\mathbf x_{kg}^{(a)}\mathbf x_{kg}^{(a)^H},\
\mathcal T_g^{(ab)} = \frac{1}{K}\sum\limits_{k = 1}^K\mathbf x_{kg}^{(a)}\mathbf x_{kg}^{(b)H}\label{cone},
\end{align}
we have
\begin{align}
\mathcal G_{aa} &= \frac{K}{T}\left(\sum\limits_{g = 1}^G\rho\mathcal C_g-\mathcal M_g^{(a)}-\mathcal M_g^{(a)H}+\mathcal Q_g^{(a)}\right)+\boldsymbol\Delta,\nonumber\\
\mathcal G_{ab} &= \frac{K}{T}\left(\sum\limits_{g = 1}^G\rho\mathcal C_g-\mathcal M_g^{(a)}-\mathcal M_g^{(b)H}+\mathcal T_g^{(ab)}\right)+\boldsymbol\Delta.\label{cov}
\end{align}
We note that the term $\beta_t$ is completely determined by the overlaps in (\ref{cone}). For further simplifying the equation (\ref{Ezn}), we introduce an identity seen in (\ref{id}), which is similarly provided in \cite{2012Probabilistic, AMPdecoder}. For concise, we use the notation
\begin{align}
{\rm d}&\boldsymbol\Theta\nonumber\\
 = &\prod_{g,a}{\rm d}\hat{\mathcal Q}_g^{(a)}{\rm d}\hat{\mathcal Q}_g^{(a)*}{\rm d}\hat{\mathcal M}_g^{(a)}{\rm d}\hat{\mathcal M}_g^{(a)*}{\rm d}\mathcal Q_g^{(a)}{\rm d}\mathcal Q_g^{(a)*}{\rm d}\mathcal M_g^{(a)}{\rm d}\mathcal M_g^{(a)*}\nonumber\\
 &\times\prod_{g,a,b}{\rm d}\hat{\mathcal T}_g^{(ab)}{\rm d}\hat{\mathcal T}_g^{(ab)*}{\rm d}\mathcal T_g^{(ab)}{\rm d}\mathcal T_g^{(ab)*}.\nonumber
\end{align}
Note that we have introduced the conjugated parameters $\{\hat{\mathcal Q}_g^{(a)}, \hat{\mathcal M}_g^{(a)}, \hat{\mathcal T}_g^{(ab)}\}$ to enforce the consistency conditions (\ref{cone}). We then consider the replica-symmetry assumption, which is a general consideration in the replica analysis. Under this, we have the following equations.
\setcounter{equation}{30}
\begin{align}
&\hat{\mathcal M}_g^{(a)} = \hat{\mathcal M}_g,\quad \hat{\mathcal Q}_g^{(a)} = \hat{\mathcal Q}_g,\quad \hat{\mathcal T}_g^{(ab)} = \hat{\mathcal T}_g,\nonumber\\
&\mathcal M_g^{(a)} = \mathcal M_g,\quad \mathcal Q_g^{(a)} = \mathcal Q_g,\quad \mathcal T_g^{(ab)} = \mathcal T_g.\label{r-s2}
\end{align}
Further combining the Hubbard-Stratonovich transform \cite{RepTan2002}
\begin{equation}
\begin{split}
&e^{\frac{1}{2}\sum\limits_{a\neq b}\boldsymbol x_{g}^{(a)H}(\hat{\mathcal T}_g+\hat{\mathcal T}_g^H)\boldsymbol x_{g}^{(b)}+\frac{1}{2}\sum\limits_{a}\boldsymbol x_{g}^{(a)H}(\hat{\mathcal T}_g+\hat{\mathcal T}_g^H)\boldsymbol x_{g}^{(a)}}\\
= &\int{\rm D}\boldsymbol ze^{\sum\limits_{a}\boldsymbol x_{g}^{(a)H}(\frac{\hat{\mathcal T}_g+\hat{\mathcal T}_g^H}{2})^{1/2}\boldsymbol z+\boldsymbol z^H\sum\limits_{a}(\frac{\hat{\mathcal T}_g+\hat{\mathcal T}_g^H}{2})^{1/2}\boldsymbol x_{g}^{(a)}},\nonumber
\end{split}
\end{equation}
the equation (\ref{Ezn}) can be re-formulated as (\ref{ezhs}).

We now turn to consider the $n\to 0$ limit. Recall the expression of $\beta_t$ in (\ref{ezhs}). By combining the equations in (\ref{cov}), the matrix $\mathcal G$ can be rewritten as
\begin{align}
\mathcal G = \boldsymbol{\amalg}_n&\otimes\left(\frac{K}{T}(\sum_g\rho\mathcal C_g-\mathcal M_g-\mathcal M_g^{H}+\mathcal T_g)+\boldsymbol\Delta\right)\nonumber\\
&+\mathbf I_n\otimes\left(\frac{K}{T}\sum_g(\mathcal Q_g-\mathcal T_g)\right),\nonumber
\end{align}
where $\boldsymbol{\amalg}_n$ stands for the $n\times n$ matrix with elements all equal to one. The eigenvalue set of $\mathcal G$ consists of the $M$ eigenvalues of matrix \begin{small}$\frac{K}{T}(\sum_g\rho\mathcal C_g-\mathcal M_g-\mathcal M_g^{H}+\mathcal T_g)+\boldsymbol\Delta$\end{small}, and $(n-1)$ groups of $M$ eigenvalues with each group consisting the eigenvalues of \begin{small}$\frac{K}{T}\sum_g(\mathcal Q_g-\mathcal T_g)$\end{small}. As a consequence, $\beta_t$ can be rewritten and further approximated in (\ref{bt}).
\setcounter{equation}{33}

Combining the final expression in (\ref{bt}) and with the approximation $$\int{\rm D}\boldsymbol zf(\boldsymbol z)^n = 1+n\int{\rm D}\boldsymbol z\log f(\boldsymbol z)\approx \exp\left( n\int{\rm D}\boldsymbol z\log f(\boldsymbol z)\right)$$ in the $n \to 0$ limit, we then obtain
\begin{equation}
\mathbb EZ^{n} = \int{\rm d}\boldsymbol\Theta \exp\left(Kn{\tilde\Phi} (\boldsymbol{\mathcal Q}, \boldsymbol{\mathcal T}, \boldsymbol{\mathcal M}, \boldsymbol{\hat{\mathcal Q}}, \boldsymbol{\hat{\mathcal T}}, \boldsymbol{\hat{\mathcal M}})\right),\label{eznf}
\end{equation}
where we have denoted $\boldsymbol{\mathcal Q}\triangleq\{\{\mathcal Q_g\}_{g = 1}^G, \{\mathcal Q^*_g\}_{g = 1}^G\}$, $\boldsymbol{\mathcal T}\triangleq\{\{\mathcal T_g\}_{g = 1}^G, \{\mathcal T^*_g\}_{g = 1}^G\}$, $\boldsymbol{\mathcal M}\triangleq\{\{\mathcal M_g\}_{g = 1}^G, \{\mathcal M^*_g\}_{g = 1}^G\}$, $\boldsymbol{\hat{\mathcal Q}}\triangleq\{\{\hat{\mathcal Q}_g\}_{g = 1}^G, \{\hat{\mathcal Q}^*_g\}_{g = 1}^G\}$, $\boldsymbol{\hat{\mathcal T}}\triangleq\{\{\hat{\mathcal T}_g\}_{g = 1}^G, \{\hat{\mathcal T}^*_g\}_{g = 1}^G\}$,
$\boldsymbol{\hat{\mathcal M}}\triangleq\{\{\hat{\mathcal M}_g\}_{g = 1}^G, \{\hat{\mathcal M}^*_g\}_{g = 1}^G\}$ as the collections of the involved variables, and the function $\tilde{\Phi}$ is specified in the equation (\ref{phi0}).

\setcounter{equation}{37}
The saddle point method \cite{AMPdecoder} is performed by taking the extremum of the (\ref{phi0}) with respect to the free parameters. Accordingly, we have
\begin{equation}
\Phi = {\rm extr}\left({\tilde\Phi} (\boldsymbol{\mathcal Q}, \boldsymbol{\mathcal T}, \boldsymbol{\mathcal M}, \boldsymbol{\hat{\mathcal Q}}, \boldsymbol{\hat{\mathcal T}}, \boldsymbol{\hat{\mathcal M}})\right),\nonumber
\end{equation}
where ${\rm extr}(\cdot)$ returns the extremum of the involved function with respect to its arguments. Setting derivatives of $\tilde{\Phi}$ with respect to $\mathcal M_g$, $\mathcal T_g$, $\mathcal Q_g-\mathcal T_g$, $\forall g$ and all their conjugations into zero, then we can obtain the self-consistent equations as follow.

\begin{align}
\forall g, \hat{\mathcal M}_g^H = &\alpha\mathcal W_g, \quad\hat{\mathcal Q}_g^H+\hat{\mathcal T}_g^H = \alpha\mathcal W_g.\nonumber\\
\forall g, \hat{\mathcal T}_g^H = &\alpha\mathcal W_g(\frac{T}{K}\boldsymbol\Delta+\sum\limits_{g}(\mathcal \rho\mathcal C_g-\mathcal M_g-\mathcal M_g^H+\mathcal T_g))\mathcal W_g,\nonumber
\end{align}
where we denote $$\mathcal W_g\triangleq (\frac{T}{K}\boldsymbol\Delta+\sum\nolimits_{g}(\mathcal Q_g-\mathcal T_g))^{-1}.$$

Under Bayes-optimal condition, we have $\mathcal Q_g = \rho\mathcal C_g$ and $\mathcal M_g = \mathcal T_g = \mathcal M_g^H$, resulting in $\hat{\mathcal M}_g^H = \hat{\mathcal T}_g^H = \alpha\mathcal W_g$ and $\hat{\mathcal Q}_g^H = \boldsymbol 0$, $\forall g$. By further considering the Nishimori identity introduced in the Sec. II-B of \cite{2012Probabilistic}, we have the relation $\mathcal EG = \sum_g \mathcal E_g = \sum_g \mathcal Q_g-\mathcal T_g$. As a consequence, after ignoring the irrelevant constant, the expression of the free entropy function with respect to $\mathcal E$ can be obtained as (\ref{general_FE}).

\subsection{Proof of Theorem 2}
\label{T2}
\setcounter{TempEqCnt}{\value{equation}}
\setcounter{equation}{35}
\begin{figure*}[btp]
\begin{align}
\Phi(\tau)& = -{\rm Tr}\left(\frac{\sum_g\rho\sigma_g^2-\tau}{\sigma_w^2+\frac{\tau}{\alpha}}\right)
-\alpha M-\alpha M\log\left(\sigma_w^2+\frac{\tau}{\alpha}\right)+\sum\limits_g\int{\rm d}\boldsymbol s_g((1-\rho)\delta(\boldsymbol s_g)+\rho Q_g(\boldsymbol s_g))\int{\rm D}\boldsymbol z\nonumber\\
&\times\log\left(\int{\rm d}\boldsymbol x_g((1-\rho)\delta(\boldsymbol x_g)+\rho Q_g(\boldsymbol x_g))\right.\left.\exp\left(-\boldsymbol x_g^H(\sigma_w^2+\frac{\tau}{\alpha})^{-1}\boldsymbol x_g+2\mathfrak R\left(\boldsymbol x_g^H\left((\sigma_w^2+\frac{\tau}{\alpha})^{-1}\boldsymbol s_g+(\sigma_w^2+\frac{\tau}{\alpha})^{-\frac{1}{2}}\boldsymbol z\right)\right)\right)\right).\label{tau1}
\end{align}
\begin{align}
\tilde{\mathcal I}_g^{(1)} &= (1-\rho)\int{\rm D}\boldsymbol z\log\left((1-\rho)+\rho\int{\rm d}\boldsymbol x_g Q_g(\boldsymbol x_g)\exp\left(-\boldsymbol x_g^H(\sigma_w^2+\frac{\tau}{\alpha})^{-1}\boldsymbol x_g+2\mathfrak R\left(\boldsymbol x_g^H(\sigma_w^2+\frac{\tau}{\alpha})^{-\frac{1}{2}}\boldsymbol z\right)\right)\right)\nonumber\\
&= (1-\rho)\int{\rm D}\boldsymbol z\log\left((1-\rho)+\rho \frac{\exp\left(||\boldsymbol z||^2\left(\sigma_w^2+\frac{\tau}{\alpha}\right)^{-1}\left(\sigma_g^{-2}+\left(\sigma_w^2+\frac{\tau}{\alpha}\right)^{-1}\right)^{-1}\right)}
{\sigma_g^{-2M}\left(\sigma_g^{-2}+\left(\sigma_w^2+\frac{\tau}{\alpha}\right)^{-1}\right)^M}\right)\nonumber\\
& = (1-\rho)\int{\rm D}\boldsymbol z\log\left((1-\rho)+\rho\left(\frac{\frac{1}{\alpha}\tau+\sigma_w^2}{\frac{1}{\alpha}\tau+\sigma_w^2+\sigma_g^2}\right)^M\exp\left(\frac{||\boldsymbol z||^2\sigma_g^2}{\sigma_g^2+\sigma_w^2+\frac{1}{\alpha}\tau}\right)\right).\label{quai1}\\
\tilde{\mathcal I}_g^{(2)} = &\rho\int{\rm D}\boldsymbol z{\rm d}\boldsymbol s_gQ_g(\boldsymbol s_g)\log\left((1-\rho)+\rho\int{\rm d}\boldsymbol x_g Q_g(\boldsymbol x_g)\exp\left(-\boldsymbol x_g^H(\sigma_w^2+\frac{\tau}{\alpha})^{-1}\boldsymbol x_g\right.\right.\left.\left.+2\mathfrak R\left(\boldsymbol x_g^H\left((\sigma_w^2+\frac{\tau}{\alpha})^{-1}\boldsymbol s_g+(\sigma_w^2+\frac{\tau}{\alpha})^{-\frac{1}{2}}\boldsymbol z\right)\right)\right)\right)\nonumber\\
 = & \rho\int{\rm D}\boldsymbol z\log\left((1-\rho)+\rho\int{\rm d}\boldsymbol x_g Q_g(\boldsymbol x_g)\exp\left(\boldsymbol x_g^H(\sigma_w^2+\frac{\tau}{\alpha})^{-1}\boldsymbol x_g+2\mathfrak R\left(\sqrt{\frac{\sigma_g^2+\sigma_w^2+\frac{\tau}{\alpha}}{\left(\sigma_w^2+\frac{\tau}{\alpha}\right)^2}}\boldsymbol x_g^H\boldsymbol z\right)\right)\right)\nonumber\\
= &\rho\int{\rm D}\boldsymbol z\log\left((1-\rho)+\rho\left(\frac{\frac{1}{\alpha}\tau+\sigma_w^2}{\frac{1}{\alpha}\tau+\sigma_w^2+\sigma_g^2}\right)^M\exp\left(\frac{\sigma_g^2}{\sigma_w^2+\frac{\tau}{\alpha}}||\boldsymbol z||^2\right)\right).\label{quai2}
\end{align}
\setcounter{equation}{38}
\begin{align}
\Phi(\mathcal E) = &-{\rm Tr}\left(\left(\boldsymbol\Delta+\frac{\mathcal EG}{\alpha}\right)^{-1}\left(\sum\limits_g\rho\mathcal C_g-\mathcal EG\right)\right)-\alpha M-\alpha \log\left|\boldsymbol\Delta+\frac{\mathcal EG}{\alpha}\right|\nonumber\\
&+\sum\limits_g\int{\rm d}\boldsymbol s_g((1-\rho)\delta(\boldsymbol s_g)+\rho Q_g(\boldsymbol s_g))\int{\rm D}\boldsymbol z\log\left((1-\rho)+\rho\frac{|\boldsymbol V_g|}{|\mathcal C_g|}\exp\{\boldsymbol\mu_g^H\boldsymbol V_g^{-1}\boldsymbol\mu_g\}\right).\label{decou0}
\end{align}
\begin{align}
\boldsymbol V_g = &\boldsymbol U_g\boldsymbol\Lambda_g\boldsymbol U_g^H-\boldsymbol U_g\boldsymbol\Lambda_g\boldsymbol U_g^H(\boldsymbol U_g\boldsymbol\Lambda_g\boldsymbol U_g^H+\boldsymbol\Delta+\frac{1}{\alpha}\tilde{\boldsymbol U}\tilde{\boldsymbol\Xi}\tilde{\boldsymbol U}^H)^{-1}\boldsymbol U_g\boldsymbol\Lambda_g\boldsymbol U_g^H\nonumber\\
 = &\boldsymbol U_g\boldsymbol\Lambda_g\boldsymbol U_g^H-\boldsymbol U_g\boldsymbol\Lambda_g\boldsymbol U_g^H(\boldsymbol U_g(\boldsymbol\Lambda_g+\sigma_w^2\mathbf I_g+\frac{1}{\alpha}\boldsymbol\Xi_g)\boldsymbol U_g^H+\sum\limits_{j\neq g}\boldsymbol U_j(\sigma_w^2\mathbf I_j+\frac{1}{\alpha}\boldsymbol\Xi_j)\boldsymbol U_j^H)^{-1}\boldsymbol U_g\boldsymbol\Lambda_g\boldsymbol U_g^H\nonumber\\
 = &\boldsymbol U_g\boldsymbol\Lambda_g\boldsymbol U_g^H-\boldsymbol U_g\boldsymbol\Lambda_g\boldsymbol U_g^H(\boldsymbol U_g(\boldsymbol\Lambda_g+\sigma_w^2\mathbf I_g+\frac{1}{\alpha}\boldsymbol\Xi_g)^{-1}\boldsymbol U_g^H+\sum\limits_{j\neq g}\boldsymbol U_j(\sigma_w^2\mathbf I_j+\frac{1}{\alpha}\boldsymbol\Xi_j)^{-1}\boldsymbol U_j^H)\boldsymbol U_g\boldsymbol\Lambda_g\boldsymbol U_g^H\nonumber\\
 = &\boldsymbol U_g\boldsymbol\Lambda_g\boldsymbol U_g^H-\boldsymbol U_g\boldsymbol\Lambda_g\boldsymbol U_g^H(\boldsymbol U_g(\boldsymbol\Lambda_g+\sigma_w^2\mathbf I_g+\frac{1}{\alpha}\boldsymbol\Xi_g)^{-1}\boldsymbol U_g^H+\sum\limits_{j\neq g}\boldsymbol U_j(\sigma_w^2\mathbf I_j)^{-1}\boldsymbol U_j^H)\boldsymbol U_g\boldsymbol\Lambda_g\boldsymbol U_g^H\nonumber\\
 = &\boldsymbol U_g\boldsymbol\Lambda_g\boldsymbol U_g^H-\boldsymbol U_g\boldsymbol\Lambda_g\boldsymbol U_g^H(\boldsymbol U_g\boldsymbol\Lambda_g\boldsymbol U_g^H+\boldsymbol\Delta+\frac{1}{\alpha}\boldsymbol U_g\boldsymbol\Xi_g\boldsymbol U_g^H)^{-1}\boldsymbol U_g\boldsymbol\Lambda_g\boldsymbol U_g^H=\mathcal C_g-\mathcal C_g(\mathcal C_g+\boldsymbol\Delta+\frac{1}{\alpha}\mathcal E_g)^{-1}\mathcal C_g.\label{decou1}
\end{align}

\hrulefill
\end{figure*}
\setcounter{equation}{\value{TempEqCnt}}
Substituting $\mathcal E = \tau G^{-1}\mathbf I$ into the general free entropy equation (7) in Sec. III-A, we can get the equation (\ref{tau1}). Do some math, we further get
$\Phi(\tau) = -\alpha M\left(\frac{\sigma_w^2}{\frac{1}{\alpha}\tau+\sigma_w^2}+\log(\sigma_w^2+\frac{1}{\alpha}\tau)\right)+\sum\nolimits_g\tilde{\mathcal I}_g^{(1)}+ \sum\nolimits_g\tilde{\mathcal I}_g^{(2)}
$, where the quantities $\tilde{\mathcal I}_g^{(1)}$ and $\tilde{\mathcal I}_g^{(2)}$ are shown in (\ref{quai1})-(\ref{quai2}).
We note that the derivation of $\tilde{\mathcal I}_g^{(2)}$ utilizes the fact that $\boldsymbol s_g$ and $\boldsymbol z$ is independent. As a consequence, we get (14) in the Sec. IV-A of the paper.
\subsection{Proof of Proposition 1}
\label{P1}
\setcounter{equation}{39}
By adopting the fact that $\exp{x}\le x+1$ for $x>0$, we get following inequalities
\[\begin{split}
&\exp\left(-\frac{\sigma_g^2}{\sigma_w^2+\frac{1}{\alpha}\tau}\right)< \frac{\frac{1}{\alpha}\tau+\sigma_w^2}{\frac{1}{\alpha}\tau+\sigma_w^2+\sigma_g^2},\\
&\exp\left(-\frac{\sigma_g^2}{\sigma_g^2+\sigma_w^2+\frac{1}{\alpha}\tau}\right)>
\frac{\frac{1}{\alpha}\tau+\sigma_w^2}{\frac{1}{\alpha}\tau+\sigma_w^2+\sigma_g^2}.
\end{split}\]
Using the law of large numbers, the term $||\boldsymbol z||^2\to M$. As a consequence, as $M\to\infty$, we have
\[\begin{split}
(1-\rho)&\exp\left(-\frac{||\boldsymbol z||^2\sigma_g^2}{\sigma_w^2+\frac{1}{\alpha}\tau}\right)+\rho\left(\frac{\frac{1}{\alpha}\tau+\sigma_w^2}{\frac{1}{\alpha}\tau+\sigma_w^2+\sigma_g^2}\right)^M\\
\approx& \rho\left(\frac{\frac{1}{\alpha}\tau+\sigma_w^2}{\frac{1}{\alpha}\tau+\sigma_w^2+\sigma_g^2}\right)^M,\\ (1-\rho)&\exp\left(-\frac{||\boldsymbol z||^2\sigma_g^2}{\sigma_g^2+\sigma_w^2+\frac{1}{\alpha}\tau}\right)+\rho\left(\frac{\frac{1}{\alpha}\tau+\sigma_w^2}{\frac{1}{\alpha}\tau+\sigma_w^2+\sigma_g^2}\right)^M\\
\approx& (1-\rho)\exp\left(-\frac{\sigma_g^2}{\sigma_g^2+\sigma_w^2+\frac{1}{\alpha}\tau}\right)^M.
\end{split}\]

As a consequence, the free entropy function (14) is reduced to (15), as we can see in the Sec. IV-A of the paper.
The first order derivative of function (15) is given by
$
\frac{\partial\Phi}{\partial \tau} = \frac{1}{\tau+\alpha\sigma_w^2}(\rho\sum_g\frac{\sigma_g^2}{\frac{1}{\alpha}\tau+\sigma_w^2+\sigma_g^2}-\frac{\tau}{\frac{1}{\alpha}\tau+\sigma_w^2})$,
and we can prove the term $\rho\sum_g\frac{\sigma_g^2}{\frac{1}{\alpha}\tau+\sigma_w^2+\sigma_g^2}$ is a monotonically decreasing with respect to $\tau$, and the term $\frac{\tau}{\frac{1}{\alpha}\tau+\sigma_w^2}$ is monotonically increasing with respect to $\tau$. Since there exists interaction in the regions of such two terms, there has and only has one point that fulfills the equation $\rho\sum_g\frac{\sigma_g^2}{\frac{1}{\alpha}\tau+\sigma_w^2+\sigma_g^2} = \frac{\tau}{\frac{1}{\alpha}\tau+\sigma_w^2}$, and the point corresponds to the global maximum of the function (15).

%\begin{small}
%\begin{equation}
%\mathcal E_g^{(\text{s})} = \mathbb E_{\hat{\boldsymbol s}_g}[(1-\hat\rho)(\mathcal C_g^{-1}+\boldsymbol\Sigma^{(\text{s})-1})^{-1}+\hat\rho(1-\hat\rho)(\mathcal C_g^{-1}+\boldsymbol\Sigma^{(\text{s})-1})^{-1}\boldsymbol\Sigma^{(\text{s})-1}\hat{\boldsymbol s}_g\hat{\boldsymbol s}_g^H\boldsymbol\Sigma^{(\text{s})-1}(\mathcal C_g^{-1}+\boldsymbol\Sigma^{(\text{s})-1})^{-1}],\label{lema1}
%\end{equation}
%\end{small}
%\begin{small}
%\begin{equation}
%(\mathcal C_g^{-1}+\boldsymbol\Sigma^{(\text{s})-1})^{-1} = \mathcal C_g-\mathcal C_g(\mathcal C_g+\boldsymbol\Sigma^{(\text{s})})^{-1}\mathcal C_g
%= \boldsymbol U_g\boldsymbol\Lambda_g\boldsymbol U_g^H-\boldsymbol U_g\boldsymbol\Lambda_g\boldsymbol U_g^H(\boldsymbol U_g\boldsymbol\Lambda_g\boldsymbol U_g^H+\boldsymbol\Sigma^{(\text{s})})^{-1}\boldsymbol U_g\boldsymbol\Lambda_g\boldsymbol U_g^H.\label{lema2}
%\end{equation}
%\end{small}
\subsection{Proof of Lemma 1}
\label{L_1}
We note that the stationary point of the free energy function (7) in the Sec. III-A with respect to every $\mathcal E_g^{(\text{s})}$ for each user group fulfills the following equation
\begin{align}
\mathcal E_g^{(\text{s})} = &\mathbb E_{\boldsymbol s_g, \boldsymbol z}\left[\left(\eta_g\left(\boldsymbol s_g+\boldsymbol\Sigma^{(\text{s})\frac{1}{2}}\boldsymbol z\right)-\boldsymbol s_g\right.\right)\nonumber\\
&\times\left(\left.\eta_g\left(\boldsymbol s_g+\boldsymbol\Sigma^{(\text{s})\frac{1}{2}}\boldsymbol z\right)-\boldsymbol s_g\right)^H\right],\nonumber
\end{align}
where we define $\boldsymbol\Sigma^{(\text{s})} \triangleq \boldsymbol\Delta+\frac{\mathcal E^{(\text{s})}G}{\alpha}$. Further denoting $\hat{\boldsymbol s}_g^{(\text{s})}\triangleq \boldsymbol s_g+\boldsymbol\Sigma^{(\text{s})\frac{1}{2}}\boldsymbol z$, we have
\begin{align}
\mathcal E_g^{(\text{s})} = &\mathbb E_{\boldsymbol s_g, \hat{\boldsymbol s}_g^{(\text{s})}}\left(\left(\eta_g\left(\hat{\boldsymbol s}_g^{(\text{s})}\right)-\boldsymbol s_g\right)
\left(\eta_g\left(\hat{\boldsymbol s}_g^{(\text{s})}\right)-\boldsymbol s_g\right)^H\right)\nonumber\\
 = &\mathbb E_{\hat{\boldsymbol s}_g^{(\text{s})}}\left(\mathbb E_{\boldsymbol s_g| \hat{\boldsymbol s}_g^{(\text{s})}}\left(\left(\eta_g\left(\hat{\boldsymbol s}_g^{(\text{s})}\right)-\boldsymbol s_g\right)
\left(\eta_g\left(\hat{\boldsymbol s}_g^{(\text{s})}\right)-\boldsymbol s_g\right)^H\right)\right).\nonumber
\end{align}
Combining the channel model in Sec. V, we have
\begin{align}
&p(\boldsymbol s_g|\hat{\boldsymbol s}_g^{(\text{s})}) =  (1-{\hat\rho})\delta(\boldsymbol s_g)\nonumber\\
+& {\hat\rho}\mathcal{CN}\big(\boldsymbol s_g; (\boldsymbol\Sigma^{(\text{s})-1}+\mathcal C_g^{-1})^{-1}\boldsymbol \Sigma^{(\text{s})-1}\hat{\boldsymbol s}_g^{(\text{s})}, (\boldsymbol\Sigma^{(\text{s})-1}+\mathcal C_g^{-1})^{-1}\big),\nonumber
\end{align}
where
$$\hat\rho \triangleq \frac{\rho\mathcal{CN}(\hat{\boldsymbol s}_g^{(\text{s})}; \boldsymbol 0, \mathcal C_g+\boldsymbol\Sigma^{(\text{s})})}{\rho\mathcal{CN}(\hat{\boldsymbol s}_g^{(\text{s})}; \boldsymbol 0, \mathcal C_g+\boldsymbol\Sigma^{(\text{s})})+(1-\rho)\mathcal{CN}(\hat{\boldsymbol s}_g; \boldsymbol 0, \boldsymbol\Sigma^{(\text{s})})}.$$
We further use the matrix inversion lemma, and we obtain $(\mathcal C_g^{-1}+\boldsymbol\Sigma^{(\text{s})-1})^{-1} = \mathcal C_g-\mathcal C_g(\mathcal C_g+\boldsymbol\Sigma^{(\text{s})})^{-1}\mathcal C_g$. Together with $\mathcal C_g = \boldsymbol U_g\boldsymbol\Lambda_g\boldsymbol U_g^H$, we finally have $\mathcal E_g^{(\text{s})} = \boldsymbol U_g\boldsymbol\Xi_g^{(\text{s})}\boldsymbol U_g^H$, with
\[\begin{split}
\boldsymbol\Xi_g^{(\text{s})}& = \mathbb E_{\hat{\boldsymbol s}_g}\left[\hat\rho(1-\hat\rho)\left(\boldsymbol\Lambda_g\boldsymbol U_g^H-\boldsymbol\Lambda_g\boldsymbol U_g^H(\mathcal C_g+\boldsymbol\Sigma^{(\text{s})})^{-1}\mathcal C_g\right)\right.\\
&\times\boldsymbol\Sigma^{(\text{s})-1}\hat{\boldsymbol s}_g\hat{\boldsymbol s}_g^H\boldsymbol\Sigma^{(\text{s})-1}\left(\boldsymbol U_g\boldsymbol\Lambda_g-\mathcal C_g(\mathcal C_g+\boldsymbol\Sigma^{(\text{s})})^{-1}\boldsymbol U_g\boldsymbol\Lambda_g\right)\\
&\left.+\hat\rho\left(\boldsymbol\Lambda_g-\boldsymbol\Lambda_g\boldsymbol U_g^H(\mathcal C_g+\boldsymbol\Sigma^{(\text{s})})^{-1}\boldsymbol U_g\boldsymbol\Lambda_g\right)\right].
\end{split}\]

\subsection{Proof of Theorem 3}
\label{T_3}
Note that by considering the specific functional form of $Q_g(\boldsymbol s_g) = \rho\mathcal {CN}(\boldsymbol s_g; \boldsymbol 0, \boldsymbol U_g\boldsymbol\Lambda_g\boldsymbol U_g^H)$, the free entropy function (7) can be rewritten as shown in the equation (\ref{decou0}), where we define
\begin{align}
\boldsymbol V_g = &\mathcal C_g-\mathcal C_g(\mathcal C_g+\boldsymbol\Delta+\frac{1}{\alpha}G\mathcal E)^{-1}\mathcal C_g, \nonumber\\ \boldsymbol\mu_g = &\boldsymbol V_g\left(\left(\boldsymbol \Delta+\frac{G\mathcal E}{\alpha}\right)^{-1}\boldsymbol s_g+\left(\boldsymbol \Delta+\frac{G\mathcal E}{\alpha}\right)^{-\frac{1}{2}}\boldsymbol z\right).\nonumber
\end{align}
Using the matrix inversion lemma together with the condition $\mathcal E =
\sum_g\boldsymbol U_g\boldsymbol\Xi_g\boldsymbol U_g^H$, we obtain the derivations in (\ref{decou1}), where the quantity $\tilde{\boldsymbol U}\triangleq [\boldsymbol U_1,\dots, \boldsymbol U_G]$ is a \emph{tall unitary} matrix and $\tilde{\boldsymbol\Xi}$ is a $\sum_gr_g$ square block diagonal matrix with $\boldsymbol\Xi_g$ as the $g$th block. Similarly, we have
$
\boldsymbol\mu_g = \boldsymbol V_g((\boldsymbol \Delta+\frac{\mathcal E_g}{\alpha})^{-1}\boldsymbol s_g+(\boldsymbol \Delta+\frac{\mathcal E_g}{\alpha})^{-\frac{1}{2}}\boldsymbol z)
$.
 As a result, we get $\Phi(\mathcal E) = \sum_{g = 0}^G\Phi_g(\mathcal E_g)$, with $\Phi_g(\mathcal E_g)$ in the form of (21) in the Sec. V-A of the paper.

\subsection{Proof of Proposition 2}\label{dp}
The definition of the detection probability is $P_D = {\rm Pr}\{\text{T}(\underline{\boldsymbol s_g})>l_g|a_g = 1\}$ with threshold $l_g$. To specify the expression of $P_D$, we define a random variable $C_{g,m}$ with one realization $$c_{g,m} = 2|\underline{s_{g,m}}|^2\left(\sigma_w^2+\frac{1}{\alpha}\xi_{g,m}^{\star}+\lambda_{g,m}\right)^{-1},$$ and as a result, we have $C_{g,m}\sim \chi^2(2)$, i.e., $p(c_{g,m})= \frac{1}{2}\exp(-\frac{1}{2}c_{g,m})$. Similarly, we further define a random variable $D_{g,m}$ with one realization $$d_{g,m} = \frac{1}{2}\lambda_{g,m}c_{g,m}\left(\sigma_w^2+\frac{1}{\alpha}\xi_{g,m}^{\star}\right)^{-1}.$$ Then, we get the PDF of random variable $D_{g,m}$ as
\begin{align}
p_{D_{g,m}}(d_{g,m}) = & (\sigma_w^2+\frac{1}{\alpha}\xi_{g,m}^{\star})\lambda_{g,m}^{-1}\nonumber\\
&\exp(-(\sigma_w^2+\frac{1}{\alpha}\xi_{g,m}^{\star})\lambda_{g,m}^{-1}d_{g,m})\mathbb I(d_{g,m}>0),\nonumber
\end{align}
where $\mathbb I(d_{g,m}>0)$ is the indicator function. Note that $p_{D_{g,m}}(d_{g,m};\omega_{g,m})$ is an exponential distribution with parameter $\omega_{g,m}\triangleq (\sigma_w^2+\frac{1}{\alpha}\xi_{g,m}^{\star})\lambda_{g,m}^{-1}$. We further assume $\omega_{g,m}$ are all distinct for each $m$ which is common in spatially correlated channel. Based on the theorem of sum of independent exponentially distributed
random variables in \cite{2013Noteson}, the PDF of the sufficient statistic with condition $a_g = 1$ can be written as
\setcounter{equation}{40}
\begin{equation}
p(\text{T}(\underline{\boldsymbol s_g})|a_g = 1) = \sum\limits_{m = 1}^{r_g}\prod_{j\neq m}\frac{\omega_{g,j}p_{D_{g,m}}(\text{T}(\underline{\boldsymbol s_g});\omega_{g,m})}{\omega_{g,j}-\omega_{g,m}},\label{pdfT}
\end{equation}
%, which is given by
%
%\begin{small}\begin{equation}
%P_D = 1-P_M = \sum\limits_{m = 1}^{r_g}\prod_{j\neq m}\frac{\omega_{g,j}}{\omega_{g,j}-\omega_{g,m}}\exp\{-\omega_{g,m}l_g\}.
%\end{equation}\end{small}
Combining the definition, the missed detection probability $P_D$ can be then obtained via the cumulative distribution function (CDF) of (\ref{pdfT}). Similarly, the sufficient statistic with condition $a_g = 0$ can be formulated as
\begin{equation}
p(\text{T}(\underline{\boldsymbol s_g})|a_g = 0) = \sum_{m = 1}^{r_g}\prod_{j\neq m}\frac{\tilde{\omega}_{g,j}p_{D_{g,m}}(\text{T}(\underline{\boldsymbol s_g});\tilde{\omega}_{g,m})}{\tilde{\omega}_{g,j}-\tilde{\omega}_{g,j}},
\end{equation}
where $\tilde{\omega}_{g,m}\triangleq (\sigma_w^2+\frac{1}{\alpha}\xi_{g,m}^{\star}+\lambda_{g,m})\lambda_{g,m}^{-1}$. As a consequence, based on the CDF of (\ref{pdfT}), the expression of false alarm probability with definition $P_F = {\rm Pr}\{\text{T}(\underline{\boldsymbol s_g})>l_g|a_g = 0\}$ can be then specified, as shown in (25) in the Sec.V-B of the paper.

\bibliography{a}
\end{document}